\documentclass[accepted=2024-11-12,a4paper]{quantumarticle}
\pdfoutput=1

\usepackage{amsmath}
\usepackage{amsthm}
\usepackage{amssymb}
\usepackage{mathtools}
\usepackage{physics}
\usepackage{graphicx}
\usepackage[pdfusetitle]{hyperref}
\usepackage[british]{babel}
\usepackage[font={small,sf},labelfont={sf,bf}]{caption}

\usepackage[utf8]{inputenc}
\usepackage{float}
\DeclareFontSeriesDefault[rm]{bf}{b}
\newcommand{\defn}[1]{\emph{#1}}
\usepackage[shortlabels]{enumitem}
\setlist{noitemsep}
\usepackage{multirow}
\usepackage{lettrine}
\usepackage{tikzit}
\usepackage{xcolor}

\tikzstyle{grijs}=[fill=none, draw=none, text=gray]

\tikzstyle{arrow}=[->, thick]
\tikzstyle{dik}=[-, thick]
\tikzstyle{grey dashed}=[-, draw={rgb,255: red,168; green,168; blue,168}, thick]
\tikzstyle{grey}=[-, draw={rgb,255: red,180; green,180; blue,180}, thick, dashed]
\tikzstyle{dik filled}=[-, thick, fill={rgb,255: red,189; green,241; blue,255}, tikzit fill={rgb,255: red,189; green,241; blue,255}]
\tikzstyle{heel dik}=[-, line width=20pt, draw=yellow, fill=yellow, rounded corners]
\tikzstyle{stippel rood}=[-, thick, densely dotted, draw=red]
\tikzstyle{stippel blauw}=[-, draw={rgb,255: red,1; green,151; blue,246}, thick, dashed]
\tikzstyle{pijl blauw}=[->, draw=blue, thick]
\tikzstyle{pijl groen}=[draw=olive, ->, thick]
\tikzstyle{dik oranje}=[-, draw={rgb,255: red,255; green,128; blue,0}, thick]
\tikzstyle{pijl grijs}=[draw=lightgray, thick, ->]
\tikzstyle{ellipsis}=[-, thick, densely dotted]
\tikzstyle{vul zwart}=[-, fill=black]

\usepackage{booktabs}
\usepackage{titlesec}

\usepackage[numbers,sort&compress]{natbib}

\newcommand{\pA}{\mathcal{A}}
\newcommand{\pB}{\mathcal{B}}
\newcommand{\pC}{\mathcal{C}}
\newcommand{\E}{\mathcal{E}}
\newcommand{\F}{\mathcal{F}}
\newcommand{\cA}{\mathcal{A}}
\newcommand{\cB}{\mathcal{B}}
\newcommand{\cC}{\mathcal{C}}
\newcommand{\GHZ}{\mathrm{GHZ}}
\DeclareMathOperator{\switch}{\textsc{switch}}
\DeclareMathOperator{\id}{id}
\DeclareMathOperator{\CHSH}{CHSH}
\DeclareMathOperator{\BC}{BC}

\renewcommand{\H}{\mathcal{H}}
\newcommand{\K}{\mathcal{K}}

\newcommand{\R}{\mathbb{R}}

\renewcommand{\a}{\alpha}
\renewcommand{\b}{\beta}

\newcommand{\la}{\lambda}
\newcommand{\ga}{\gamma}
\let\phi\varphi
\newcommand{\tns}{\otimes}
\usepackage{bm}
\DeclareMathAlphabet{\mathbfit}{OT1}{lmr}{b}{sl}
\let\vec\mathbfit
\newcommand{\implshort}[1]{\ensuremath{\Longrightarrow_{#1}}}
\newcommand{\impl}[1]{\ \ \implshort{#1}\ \ }

\theoremstyle{definition}
\newtheorem{definition}{Definition}
\newtheorem{assumption}{Assumption}
\newtheorem{fact}{Observation}

\theoremstyle{plain}
\newtheorem{theorem}{Theorem}
\newtheorem*{theorem*}{Theorem}
\newtheorem{lemma}{Lemma}
\newtheorem*{claim*}{Claim}
\newtheorem{corollary}{Corollary}

\newcommand{\indep}{\mathrel{\raisebox{0.05em}{\rotatebox[origin=c]{90}{$\models$}}}}

\begin{document}

    \title{Possibilistic and maximal indefinite causal order in the quantum switch}
    \author{Tein van der Lugt}
    \affiliation{Department of Computer Science, University of Oxford, Wolfson Building, Parks Road, Oxford OX1 3QD, United Kingdom}
    \email{tein.van.der.lugt@cs.ox.ac.uk}
    \author{Nick Ormrod}
    \affiliation{Department of Computer Science, University of Oxford, Wolfson Building, Parks Road, Oxford OX1 3QD, United Kingdom}
    \affiliation{Perimeter Institute for Theoretical Physics, 31 Caroline Street North, Waterloo, Ontario N2L 2Y5, Canada}
    \email{normrod@perimeterinstitute.ca}

    \begin{abstract}
        \noindent It was recently found that the indefinite causal order in the quantum switch can be certified device-independently when assuming the impossibility of superluminal influences.
        Here we strengthen this result in two ways.
        First, we give a proof of this fact which is possibilistic rather than probabilistic, i.e.\ which does not rely on the validity of probability theory at the hidden variable level.
        Then, returning to the probabilistic setting, we show that the indefinite causal order in the quantum switch is also maximal, in the sense that the observed correlations are incompatible even with the existence of a causal order on only a small fraction of the runs of the experiment.
        While the original result makes use of quantum theory's violation of a Clauser-Horne-Shimony-Holt inequality, the proofs presented here are based on Greenberger, Horne, and Zeilinger's and Mermin's proofs of nonlocality, respectively.
    \end{abstract}

    \maketitle

    When two interventions are performed within the context of a quantum circuit, the topology of the circuit puts constraints on the causal relations between them: either the choice of operation at $A$ can influence the outcome of the operation at $B$, or vice versa, but not both. Recent extensions to the quantum circuit paradigm have however facilitated considering situations in which the topology of the quantum circuit itself, and in particular the causal order between operations, is controlled coherently. The simplest example of this is the \emph{quantum switch}~\cite{CDPV13}, which applies two operations to a target system in an order coherently controlled by the state of a qubit.

It was recently found that under some metaphysical assumptions, the indefinite causal order between operations in the quantum switch can also be demonstrated device-independently, that is, relying just on correlations between classical settings and outcomes of the operations, and not on a characterisation of the operations themselves (or indeed the assumption that they are governed by quantum theory)~\cite{vdLBC23,GP23,DASB24}.
The argument of~\cite{vdLBC23} works via the violation of an inequality derived from assumptions named definite causal order, relativistic causality, and free interventions.
This derivation relies closely on Bell's theorem: the crucial observation is that a hidden variable determining a causal order (which exists by the first assumption) in some particular cases must also fix, or, in other words, `predetermine', a measurement outcome---and that this determinism, together with the other assumptions, implies Bell inequalities which are violated by the quantum correlations under consideration.

This observation suggests that it is possible to transform other proofs of Bell nonlocality, beyond that based on the Clauser-Horne-Shimony-Holt (CHSH) inequality employed in~\cite{vdLBC23}, into proofs of indefinite causal order in the quantum switch under assumptions of relativistic causality and free interventions.
Here we give a proof in the style of Mermin's version of Greenberger-Horne-Zeilinger (GHZ) nonlocality~\cite{GHZ89, Mer90, GHSZ90}.
Similarly to the latter, the proof does not involve inequalities and is of a possibilistic, rather than probabilistic, character:
the quantum switch data are shown to be incompatible with definite causal order, relativistic causality, and free interventions even when these notions are formulated merely in terms of which events have zero and nonzero probabilities.
In particular, this shows that the quantum switch is incompatible with a relativistically well-behaved definite causal order even if causal orders are not required to follow the laws of probability theory.

We also give a statistical inequality, analogous to Mermin's inequality~\cite{Mer90a} for Bell nonlocality, which parallels our possibilistic argument.
Although this takes us back to assuming the validity of probability theory, it allows us to improve on the results of~\cite{vdLBC23,GP23} in a different way, namely by showing that the quantum switch exhibits `maximal' indefinite causal order.
More precisely, the correlations obtained from the quantum switch violate this inequality to the algebraic maximum, which implies incompatibility even with the hypothesis that there is a definite causal order on only some positive fraction of the runs of the experiment.

    \section{Notation and terminology}
    \label{sec:notation}
    The first part of this paper relies only on statements about possibilities, rather than about probabilities.
Possibilities are modelled by the Boolean semiring ${\mathbb{B} = \{0,1\}}$, where $0$ denotes impossibility and $1$ possibility of an event.
Multiplication and addition in $\mathbb{B}$ are defined as $0\cdot 0 = 0\cdot 1=1\cdot 0=0$, $1\cdot 1 = 1$, $0+0=0$, $0+1=1+0=1$, and $1+1 = 1$.
A \defn{possibility distribution} over a set of variables $a_1,\dots,a_n$ taking values in the finite sets $A_1,\dots,A_n$ is a function
$p: A_1\times\cdots\times A_n \to \mathbb{B}$ satisfying
$\sum_{a_1,\dots,a_n} p(a_1 \cdots a_n) = 1$.
Note that summation here is in $\mathbb B$, where $1+1=1$; thus, this expression means that at least one tuple of values is possible.
(We will abuse notation by using lower-case letters to denote random variables as well as their values.)
A \emph{probability} distribution $P:A\to\R_{\geq0}$ induces a possibility distribution $p$ over the same variable by $p(a) = \pi(P(a))$, where $\pi:\R_{\geq 0} \to \mathbb B$ takes $0$ to $0$ and any $r > 0$ to $1$.

(Conditional) independence of possibilities is denoted by $\indep_{(\cdot)}$.
In the case of a possibility distribution $p(abc)$, $a\indep_p b \mid c$ is defined as $\forall a,b,c: p(abc) = p(ac)p(bc)$, while unconditional independence $a\indep_p b$ means that $\forall a,b : p(ab)=p(a)p(b)$.
This essentially expresses that if the values $a,b$ are both possible (together with $c$), then they are also jointly possible; note that this is strictly weaker than probabilistic independence.
This notation is discussed in more detail in Appendix~\ref{app:independence-notation}.

The arrow $\Longrightarrow$ denotes implication: for propositions (i.e.\ binary variables) $P$ and $Q$, $P\Longrightarrow_p Q$ is short for $p(P,\neg Q) = 0$.
Finally, $\oplus$ denotes addition modulo 2, $\ket{\pm} \coloneqq {(\ket0 \pm \ket1)/\sqrt2}$, and $\ket{\pm i} \coloneqq (\ket0 \pm i\ket1)/\sqrt2$.

    \section{GHZ-Mermin nonlocality}
    \label{sec:vanilla-ghz}
    We first review Mermin's version~\cite{Mer90} of the GHZ proof~\cite{GHZ89} of Bell nonlocality.
Three spacelike-separated parties labelled $A$, $B$, and $C$ share three qubits prepared in the state $\ket\GHZ \coloneqq (\ket{000}+\ket{111})/\sqrt2$.
Each party has a binary classical input variable ($x,y,z\in\{0,1\}$, resp.), and measures their qubit in the $Y$ basis $\{\ket{+i},\ket{-i}\}$ if their input is 0, and in the $X$ basis $\{\ket+,\ket-\}$ if their input is 1.
The measurement outcomes are recorded in the output variables $a,b,c\in\{0,1\}$.
Given a probability distribution over the input variables, the Born rule provides joint probabilities $P(abcxyz)$ for all combinations of input and output variables. All that the following argument relies on are however the induced \emph{possibilities} $p(abcxyz) \coloneqq \pi(P(abcxyz))$.%
\footnote{Conditional distribution notation like $P(abc|xyz)$ is more common in the literature on Bell nonlocality and indefinite causal order. We adopt a convention in which all variables are treated equal, both because this requires fewer a priori assumptions (e.g.\ about whether all combinations of settings are possible, see Eq.~\eqref{eq:all-inputs-possible}), and because it is more natural in the possibilistic view: $x$, like $a$, is a variable that takes one and only one value on each single trial of the experiment, and possibility distributions tell us which values of $x$ and $a$ can occur jointly in a single trial. See also Appendix~\ref{app:independence-notation}.}

A particular property of these possibilities, predicted by the Born rule, is that
\begin{equation}\label{eq:ghz-cors}
    \begin{split}
        x=1,y=1,z=1 \impl{p} a\oplus b \oplus c = 0; \\
        x=1,y=0,z=0 \impl{p} a\oplus b \oplus c = 1; \\
        x=0,y=1,z=0 \impl{p} a\oplus b \oplus c = 1; \\
        \text{ and\quad} x=0,y=0,z=1 \impl{p} a\oplus b \oplus c = 1.
    \end{split}
\end{equation}
These perfect correlations can be used in an EPR-like argument~\cite{EPR35} to argue for the existence of local deterministic hidden variables, i.e.\ to argue that at each wing of the experiment, there are pre-existing physical properties that determine the outcome for any possible measurement setting on that wing. Mathematically, these are described by variables $\la_{A,B,C}^{0,1}\in\{0,1\}$, independent of $x,y,z$ and satisfying $x = i \implshort{p} a = \la_A^i$ and similar conditions for $B$ and $C$.

Together with the assumption that all combinations of inputs appearing in~\eqref{eq:ghz-cors} are in fact possible, this leads us to conclude that with certainty
\begin{equation}\label{eq:ghz-hv-cors}
    \begin{split}
        \la_A^1 \oplus \la_B^1 \oplus \la_C^1 &= 0; \\
        \la_A^1 \oplus \la_B^0 \oplus \la_C^0 &= 1; \\
        \la_A^0 \oplus \la_B^1 \oplus \la_C^0 &= 1; \\
        \la_A^0 \oplus \la_B^0 \oplus \la_C^1 &= 1; \quad \oplus \\
        \noalign{\smallskip} \hline \noalign{\smallskip}
        0 &= 1.
    \end{split}
\end{equation}
This contradiction demonstrates that no model with local deterministic hidden variables $\la_{A,B,C}^{0,1}$ can reproduce the correlations~\eqref{eq:ghz-cors} predicted by the Born rule.


    \section{Three GHZ-controlled quantum switches}
    \label{sec:switch-data}
    We will now see how this argument can be turned into a possibilistic proof of indefinite causal order in the quantum switch. Our first objective will be to identify a scenario involving quantum switches that yields possibilistic data satisfying properties analogous to~\eqref{eq:ghz-cors} above.

The quantum switch~\cite{CDPV13} is defined as a supermap~\cite{CDP08} taking two quantum operations $\E,\F$ on a \emph{target system} $T$ to an operation $\switch(\E,\F)$ on the joint system $CT$, which applies $\E$ and $\F$ to $T$ in an order that is coherently controlled by the state of the \emph{control qubit} $C$ (see Figure~\ref{fig:switch-supermap}). We will consider a variant of the quantum switch which is controlled in the $X$ basis $\{\ket+,\ket-\}$. Hence, if the target and control systems are described by Hilbert spaces $\H_T$ and $\H_C\cong\mathbb C^2$ and if $\E(\cdot)=E(\cdot)E^\dagger$ and $\F(\cdot)=F(\cdot)F^\dagger$ are pure operations with Kraus operators $E,F:\H_T\to\H_T$, then $\switch(\E,\F)(\cdot)=W(\cdot)W^\dagger$ where $W:\H_C\tns\H_T\to\H_C\tns\H_T$ is the operator defined by
\begin{equation}
    \label{eq:switch-defn}
    W \coloneqq \ket+\bra+_C \tns FE + \ket-\bra-_C \tns EF.
\end{equation}

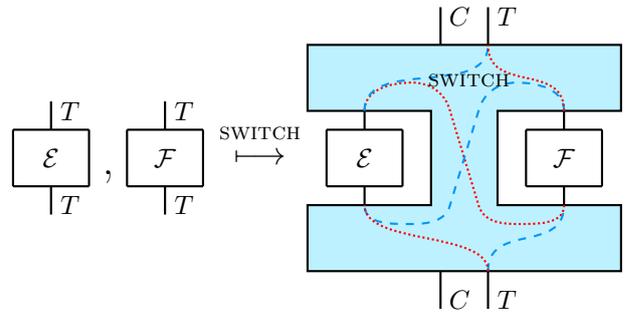
\begin{figure}
    \begin{tikzpicture}
	\begin{pgfonlayer}{nodelayer}
		\node [style=none] (0) at (-5, 1.5) {};
		\node [style=none] (1) at (-5, -0.25) {};
		\node [style=none] (2) at (-1.75, -0.25) {};
		\node [style=none] (3) at (0, -0.25) {};
		\node [style=none] (4) at (3.25, -0.25) {};
		\node [style=none] (5) at (3.25, 1.5) {};
		\node [style=none] (6) at (-1.75, -2.75) {};
		\node [style=none] (7) at (-5, -2.75) {};
		\node [style=none] (8) at (-5, -4.5) {};
		\node [style=none] (9) at (0, -2.75) {};
		\node [style=none] (10) at (3.25, -2.75) {};
		\node [style=none] (11) at (3.25, -4.5) {};
		\node [style=none] (12) at (-3.5, -0.75) {};
		\node [style=none] (13) at (-3.5, -2.25) {};
		\node [style=none] (14) at (1.75, -0.75) {};
		\node [style=none] (15) at (1.75, -2.25) {};
		\node [style=none] (16) at (-3.5, -0.25) {};
		\node [style=none] (17) at (-3.5, -2.75) {};
		\node [style=none] (18) at (1.75, -2.75) {};
		\node [style=none] (19) at (1.75, -0.25) {};
		\node [style=none] (20) at (-0.25, 1.5) {};
		\node [style=none] (21) at (-0.25, 2.5) {};
		\node [style=none] (22) at (-1.5, 1.5) {};
		\node [style=none] (23) at (-1.5, 2.5) {};
		\node [style=none] (27) at (-1.5, -4.5) {};
		\node [style=none] (28) at (-0.25, -4.5) {};
		\node [style=none] (29) at (-1.5, -5.5) {};
		\node [style=none] (30) at (-0.25, -5.5) {};
		\node [style=none] (38) at (0.25, 2.25) {$T$};
		\node [style=none] (39) at (-1, 2.25) {$C$};
		\node [style=none] (40) at (-1, -5.25) {$C$};
		\node [style=none] (41) at (0.25, -5.25) {$T$};
		\node [style=none] (54) at (1.75, -0.75) {};
		\node [style=none] (55) at (1.75, -2.25) {};
		\node [style=none] (56) at (-4.5, -0.75) {};
		\node [style=none] (57) at (-2.5, -0.75) {};
		\node [style=none] (58) at (-2.5, -2.25) {};
		\node [style=none] (59) at (-4.5, -2.25) {};
		\node [style=none] (60) at (-3.5, -1.5) {$\mathcal E$};
		\node [style=none] (61) at (0.75, -0.75) {};
		\node [style=none] (62) at (2.75, -0.75) {};
		\node [style=none] (63) at (2.75, -2.25) {};
		\node [style=none] (64) at (0.75, -2.25) {};
		\node [style=none] (70) at (-0.75, 0.55) {\small\textsc{switch}};
		\node [style=none] (71) at (-3.5, -0.25) {};
		\node [style=none] (72) at (-3.5, -2.75) {};
		\node [style=none] (73) at (1.75, -2.75) {};
		\node [style=none] (74) at (1.75, -0.25) {};
		\node [style=none] (75) at (-2.25, 0.5) {};
		\node [style=none] (76) at (-1.25, -0.25) {};
		\node [style=none] (77) at (-0.5, -2.75) {};
		\node [style=none] (78) at (0.5, -3.25) {};
		\node [style=none] (79) at (0.5, 0.5) {};
		\node [style=none] (80) at (-0.5, -0.25) {};
		\node [style=none] (81) at (-1.25, -2.75) {};
		\node [style=none] (82) at (-2.25, -3.25) {};
		\node [style=none] (87) at (1.75, -1.5) {$\mathcal F$};
		\node [style=none] (88) at (-11.75, -0.75) {};
		\node [style=none] (89) at (-11.75, -2.25) {};
		\node [style=none] (90) at (-12.75, -0.75) {};
		\node [style=none] (91) at (-10.75, -0.75) {};
		\node [style=none] (92) at (-10.75, -2.25) {};
		\node [style=none] (93) at (-12.75, -2.25) {};
		\node [style=none] (94) at (-11.75, -1.5) {$\mathcal E$};
		\node [style=none] (95) at (-11.25, -0.25) {$T$};
		\node [style=none] (96) at (-11.75, -3) {};
		\node [style=none] (97) at (-11.75, 0) {};
		\node [style=none] (98) at (-11.25, -2.75) {$T$};
		\node [style=none] (99) at (-10.25, -2) {\LARGE ,};
		\node [style=none] (100) at (-8.75, -0.75) {};
		\node [style=none] (101) at (-8.75, -2.25) {};
		\node [style=none] (102) at (-9.75, -0.75) {};
		\node [style=none] (103) at (-7.75, -0.75) {};
		\node [style=none] (104) at (-7.75, -2.25) {};
		\node [style=none] (105) at (-9.75, -2.25) {};
		\node [style=none] (106) at (-8.75, -1.5) {$\mathcal F$};
		\node [style=none] (107) at (-8.25, -0.25) {$T$};
		\node [style=none] (108) at (-8.75, -3) {};
		\node [style=none] (109) at (-8.75, 0) {};
		\node [style=none] (110) at (-8.25, -2.75) {$T$};
		\node [style=none] (111) at (-6.25, -1.5) {\large $\longmapsto$};
		\node [style=none] (112) at (-6.25, -0.85) {\small\textsc{switch}};
	\end{pgfonlayer}
	\begin{pgfonlayer}{edgelayer}
		\draw [style=dik filled] (3.center)
			 to (4.center)
			 to (5.center)
			 to (0.center)
			 to (1.center)
			 to (2.center)
			 to (6.center)
			 to (7.center)
			 to (8.center)
			 to (11.center)
			 to (10.center)
			 to (9.center)
			 to cycle;
		\draw [style=dik] (16.center) to (12.center);
		\draw [style=dik] (13.center) to (17.center);
		\draw [style=dik] (15.center) to (18.center);
		\draw [style=dik] (19.center) to (14.center);
		\draw [style=dik] (21.center) to (20.center);
		\draw [style=dik] (22.center) to (23.center);
		\draw [style=dik, in=90, out=-90] (28.center) to (30.center);
		\draw [style=dik, in=90, out=-90, looseness=0.75] (27.center) to (29.center);
		\draw [style=dik] (59.center) to (56.center);
		\draw [style=dik] (56.center) to (57.center);
		\draw [style=dik] (57.center) to (58.center);
		\draw [style=dik] (58.center) to (59.center);
		\draw [style=dik] (64.center) to (61.center);
		\draw [style=dik] (61.center) to (62.center);
		\draw [style=dik] (62.center) to (63.center);
		\draw [style=dik] (63.center) to (64.center);
		\draw [style=stippel rood] (73.center)
			 to [in=0, out=-90] (78.center)
			 to [in=-75, out=-180] (77.center)
			 to (76.center)
			 to [in=0, out=105] (75.center)
			 to [in=90, out=180] (71.center);
		\draw [style=stippel blauw] (74.center)
			 to [in=0, out=90] (79.center)
			 to [in=75, out=180] (80.center)
			 to (81.center)
			 to [in=0, out=-105] (82.center)
			 to [in=-90, out=-180] (72.center);
		\draw [style=stippel blauw, in=-90, out=90] (71.center) to (20.center);
		\draw [style=stippel rood, in=-90, out=90, looseness=1.25] (74.center) to (20.center);
		\draw [style=stippel blauw, in=-90, out=90, looseness=1.25] (28.center) to (73.center);
		\draw [style=stippel rood, in=-90, out=90, looseness=0.75] (28.center) to (72.center);
		\draw [style=dik] (93.center) to (90.center);
		\draw [style=dik] (90.center) to (91.center);
		\draw [style=dik] (91.center) to (92.center);
		\draw [style=dik] (92.center) to (93.center);
		\draw [style=dik] (97.center) to (88.center);
		\draw [style=dik] (89.center) to (96.center);
		\draw [style=dik] (105.center) to (102.center);
		\draw [style=dik] (102.center) to (103.center);
		\draw [style=dik] (103.center) to (104.center);
		\draw [style=dik] (104.center) to (105.center);
		\draw [style=dik] (109.center) to (100.center);
		\draw [style=dik] (101.center) to (108.center);
	\end{pgfonlayer}
\end{tikzpicture}
    \caption{The quantum switch takes two quantum operations on the system $T$, here $\E$ and $\F$, to an operation on $CT$, where $C$ is a control qubit. The dotted (red) and dashed (blue) lines illustrate the wirings to which the quantum switch reduces upon preparation of $C$ in state $\op+$ and $\op-$, respectively.}
    \label{fig:switch-supermap}
\end{figure}

We consider a scenario (see Figure~\ref{fig:switchy-ghz}) wherein three such quantum switches, labelled by $A$, $B$ and $C$, are implemented at spacelike separation.%
\footnote{See Appendix~\ref{app:three-switches} for an explanation of why just one quantum switch does not suffice.}
Their control qubits are prepared in the GHZ state, while their target systems, also qubits, are all prepared in the state $\ket0$.
We will think of the two operations on the target system $T_A$ inside switch $A$ as being performed by two agents, $\pA_1$ and $\pA_2$, with classical input variables $x_1,x_2\in\{0,1\}$ and output variables $a,b\in\{0,1\}$, respectively.
If $x_i=0$ (for $i=1,2$), then $\pA_i$ performs no intervention (i.e.\ lets $T_A$ undergo the identity channel), and outputs $a_i=0$.
If $x_i=1$, she instead measures $T_A$ in the computational basis, records the result in $a_i\in\{0,1\}$, and prepares the outgoing target system in state $\ket{1}_{T_A}$.%
\footnote{Other choices of operations, like the ones used in Ref.~\cite{vdLBC23}, are also possible, and would lead to similar results.}
A third agent $\pA_3$ has no input variable and always measures the output control system $C_A$ in the $Y$ basis, recording their result in the output variable $a_3\in\{0,1\}$.
The output target system is discarded.
Similar operations are performed on the systems of switches $B$ and $C$.

The Born rule, together with the definition of the quantum switch~\eqref{eq:switch-defn} and a distribution over input variables, provides us with a probability distribution $Q(a_1a_2a_3b_1b_2b_3c_1c_2c_3x_1x_2y_1y_2z_1z_2)$ (see Appendix~\ref{app:switch-data} for a more precise definition). We will abbreviate this by $Q(\vec a \vec b\vec c\vec x \vec y \vec z)$ where $\vec a \coloneqq (a_1,a_2,a_3)$, $\vec x \coloneqq (x_1, x_2)$ and $\vec b$, $\vec c$, $\vec y$, and $\vec z$ are defined similarly.
As was the case for GHZ nonlocality, it will suffice to consider only the induced possibility distribution $q(\vec a \vec b\vec c \vec x \vec y \vec z) \coloneqq \pi(Q(\vec a \vec b\vec c\vec x \vec y \vec z))$.

\begin{figure*}
    \centering
    \input{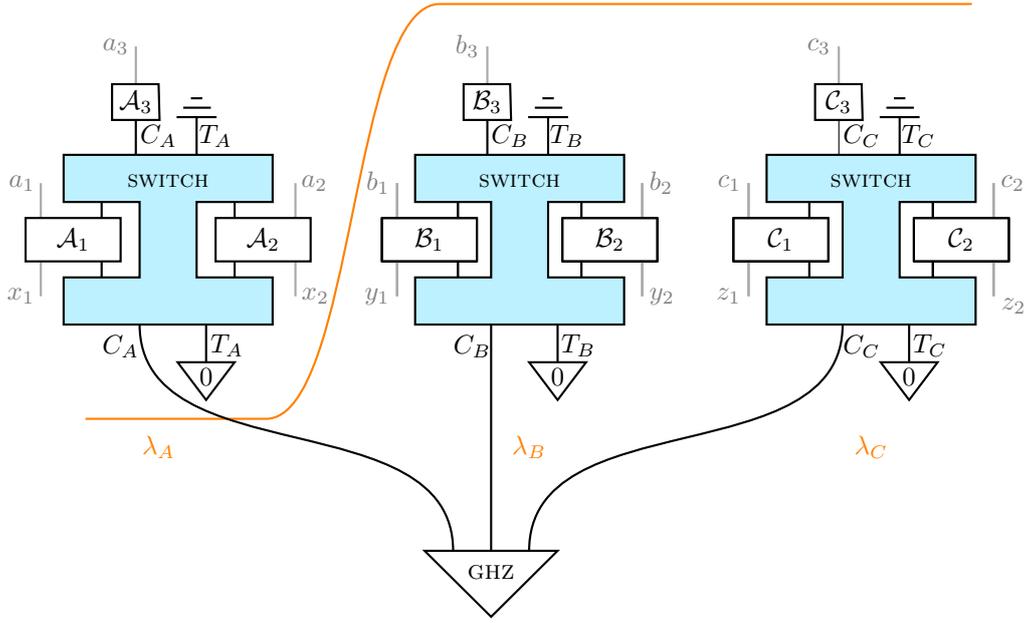}
    \caption{A GHZ-inspired quantum switch setup.
    The diagram is read from bottom to top; black wires are quantum systems while grey wires are classical variables.
    Three quantum switches (in blue) have control qubits entangled in the GHZ state.
    Their target systems are qubits, initially prepared in the computational basis state $\ket0$.
    Three parties act on each switch; their actions, here represented by classical-quantum channels, are described in the main text.
    The output target systems are discarded.
    Overall, this diagram defines the classical channel (conditional probability distribution) $Q(\vec a \vec b \vec c\mid \vec x \vec y \vec z)$.
    The orange variables $\la_{A,B,C}$ denote the postulated hidden causal orders, which are assumed to take a value in the past lightcone of their respective switches. The orange line is an example of a spacelike hypersurface that can be used to argue for Equation~\eqref{eq:locality-2}.
    }
    \label{fig:switchy-ghz}
\end{figure*}

We now want to find patterns in these possibilities analogous to those in~\eqref{eq:ghz-cors}. For that purpose, the following two observations will be useful (cf.\ Figure~\ref{fig:rmk-1-and-2}). (These informal statements are given to provide intuition; formalised versions and their proofs can be found as Lemmas~\ref{lma:obs-1} and~\ref{lma:obs-2} in Appendix~\ref{app:switch-data}.)
\begin{fact}
    \label{fct:x-are-0}
    If $x_1=x_2=0$, then $a_3$ simulates the outcome of a $Y$ basis measurement of the \emph{input} control system $C_A$; similarly for switches $B$ and $C$.
\end{fact}
Indeed, when $x_1=x_2=0$, no interventions are performed on the target system inside the switch, so that the switch itself reduces to the identity channel. This leaves the control system unaffected for $\pA_3$ to measure (see Figure~\ref{fig:rmk-1-and-2}a).

\begin{fact}
    \label{fct:x-are-1}
    If $x_1=x_2=1$, then $a_1$ simulates the outcome of an $X$ basis measurement of the input control system $C_A$; similarly for switches $B$ and $C$.
\end{fact}
This is because the input target system is initially prepared in state $\ket0_{T_A}$, while $\pA_2$ (given $x_2=1$) reprepares it in the orthogonal state $\ket1_{T_A}$. As a result, a measurement of $a_1=1$ is only compatible with the wiring in which $\pA_1$ comes after $\pA_2$---and therefore only with finding the control qubit in state $\ket-_{C_A}$ upon an $X$ basis measurement. Similarly, $a_1=0$ is only compatible with the control qubit being found in state $\ket+_{C_A}$. Therefore the probabilities (and hence possibilities) of these events are identical.
This is depicted diagrammatically in Figure~\ref{fig:rmk-1-and-2}b and proven in Appendix~\ref{app:switch-data}.


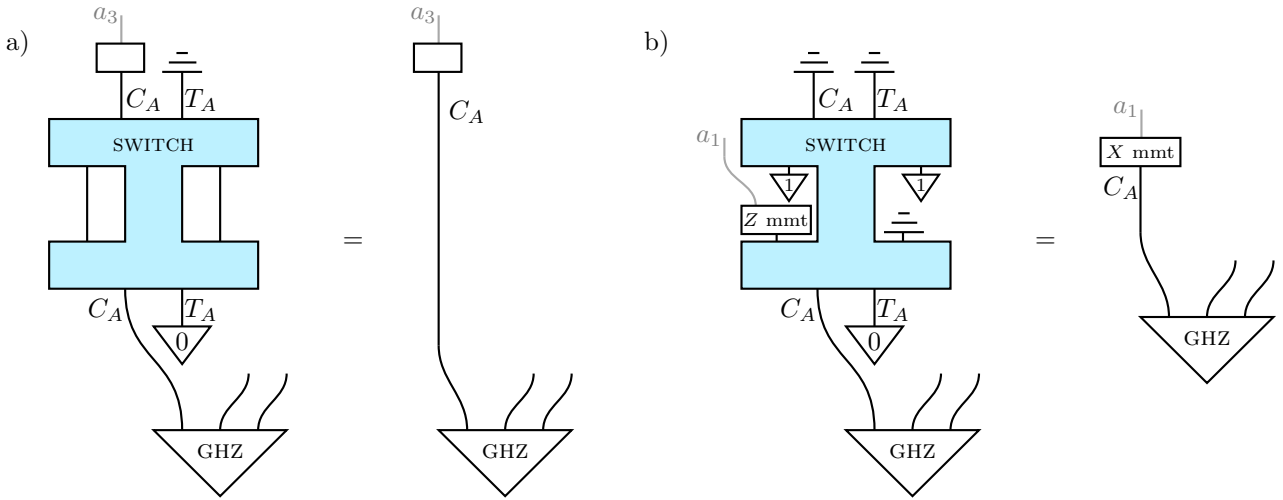
\begin{figure*}
    \centering
    \scalebox{1}{a) \begin{tikzpicture}
	\begin{pgfonlayer}{nodelayer}
		\node [style=none] (0) at (2.25, -2) {};
		\node [style=none] (1) at (2.25, -3.25) {};
		\node [style=none] (2) at (4.25, -3.25) {};
		\node [style=none] (3) at (5.75, -3.25) {};
		\node [style=none] (4) at (7.75, -3.25) {};
		\node [style=none] (5) at (7.75, -2) {};
		\node [style=none] (6) at (4.25, -5.25) {};
		\node [style=none] (7) at (2.25, -5.25) {};
		\node [style=none] (8) at (2.25, -6.5) {};
		\node [style=none] (9) at (5.75, -5.25) {};
		\node [style=none] (10) at (7.75, -5.25) {};
		\node [style=none] (11) at (7.75, -6.5) {};
		\node [style=none] (13) at (3.25, -3.25) {};
		\node [style=none] (15) at (6.75, -3.25) {};
		\node [style=none] (17) at (3.25, -5.25) {};
		\node [style=none] (18) at (6.75, -5.25) {};
		\node [style=none] (20) at (5.75, -2) {};
		\node [style=none] (21) at (5.75, -0.75) {};
		\node [style=none] (22) at (4.15, -2) {};
		\node [style=none] (23) at (4.15, -0.75) {};
		\node [style=none] (24) at (4.25, -6.5) {};
		\node [style=none] (25) at (5.75, -6.5) {};
		\node [style=none] (26) at (5.75, -7.5) {};
		\node [style=none] (27) at (6.375, -1.5) {$T_A$};
		\node [style=none] (28) at (4.775, -1.5) {$C_A$};
		\node [style=none] (29) at (4.9, -7) {$C_A$};
		\node [style=none] (30) at (6.325, -7) {$T_A$};
		\node [style=none] (31) at (5, -2.7) {\small$\textsc{switch}$};
		\node [style=none] (43) at (5.75, -0.75) {};
		\node [style=none] (44) at (5.225, -0.75) {};
		\node [style=none] (45) at (6.275, -0.75) {};
		\node [style=none] (46) at (5.4, -0.5) {};
		\node [style=none] (47) at (6.125, -0.5) {};
		\node [style=none] (48) at (5.6, -0.25) {};
		\node [style=none] (49) at (5.9, -0.25) {};
		\node [style=none] (50) at (5, -7.5) {};
		\node [style=none] (51) at (6.5, -7.5) {};
		\node [style=none] (52) at (5.75, -8.5) {};
		\node [style=none] (53) at (5.75, -7.9) {0};
		\node [style=none] (55) at (5.75, -10.25) {};
		\node [style=none] (56) at (5.75, -10.25) {};
		\node [style=none] (57) at (5, -10.25) {};
		\node [style=none] (58) at (8.5, -10.25) {};
		\node [style=none] (59) at (6.75, -12) {};
		\node [style=none] (60) at (6.75, -10.25) {};
		\node [style=none] (61) at (7.75, -10.25) {};
		\node [style=none] (62) at (6.75, -10.825) {\textsc{ghz}};
		\node [style=none] (63) at (7.5, -8.75) {};
		\node [style=none] (64) at (8.5, -8.75) {};
		\node [style=none] (67) at (10.25, -5.25) {\large$=$};
		\node [style=none] (68) at (13.2, -1.5) {$C_A$};
		\node [style=none] (69) at (13.25, -10.25) {};
		\node [style=none] (70) at (13.25, -10.25) {};
		\node [style=none] (71) at (12.5, -10.25) {};
		\node [style=none] (72) at (16, -10.25) {};
		\node [style=none] (73) at (14.25, -12) {};
		\node [style=none] (74) at (14.25, -10.25) {};
		\node [style=none] (75) at (15.25, -10.25) {};
		\node [style=none] (76) at (14.25, -10.825) {\textsc{ghz}};
		\node [style=none] (77) at (15, -8.75) {};
		\node [style=none] (78) at (16, -8.75) {};
		\node [style=none] (81) at (12.5, -8) {};
		\node [style=none] (186) at (12.5, -0.75) {};
		\node [style=none] (189) at (3.5, -0.75) {};
		\node [style=none] (190) at (4.75, -0.75) {};
		\node [style=none] (191) at (3.5, 0) {};
		\node [style=none] (192) at (4.75, 0) {};
		\node [style=none] (193) at (4.15, 0) {};
		\node [style=none] (194) at (4.15, 0.75) {};
		\node [style=grijs] (195) at (3.75, 0.75) {$a_3$};
		\node [style=none] (196) at (12.5, -0.75) {};
		\node [style=none] (197) at (11.85, -0.75) {};
		\node [style=none] (198) at (13.1, -0.75) {};
		\node [style=none] (199) at (11.85, 0) {};
		\node [style=none] (200) at (13.1, 0) {};
		\node [style=none] (201) at (12.5, 0) {};
		\node [style=none] (202) at (12.5, 0.75) {};
		\node [style=grijs] (203) at (12.1, 0.75) {$a_3$};
	\end{pgfonlayer}
	\begin{pgfonlayer}{edgelayer}
		\draw [style=grey dashed] (193.center) to (194.center);
		\draw [style=dik filled] (3.center)
			 to (4.center)
			 to (5.center)
			 to (0.center)
			 to (1.center)
			 to (2.center)
			 to (6.center)
			 to (7.center)
			 to (8.center)
			 to (11.center)
			 to (10.center)
			 to (9.center)
			 to cycle;
		\draw [style=dik] (13.center) to (17.center);
		\draw [style=dik] (15.center) to (18.center);
		\draw [style=dik] (21.center) to (20.center);
		\draw [style=dik] (22.center) to (23.center);
		\draw [style=dik, in=90, out=-90] (25.center) to (26.center);
		\draw [style=dik] (48.center) to (49.center);
		\draw [style=dik] (46.center) to (47.center);
		\draw [style=dik] (44.center) to (45.center);
		\draw [style=dik] (51.center)
			 to (52.center)
			 to (50.center)
			 to cycle;
		\draw [style=dik, in=90, out=-90, looseness=1.25] (24.center) to (55.center);
		\draw [style=dik] (58.center)
			 to (59.center)
			 to (57.center)
			 to cycle;
		\draw [style=dik, in=-90, out=90] (60.center) to (63.center);
		\draw [style=dik, in=-90, out=90] (61.center) to (64.center);
		\draw [style=dik] (71.center)
			 to (72.center)
			 to (73.center)
			 to cycle;
		\draw [style=dik, in=-90, out=90] (74.center) to (77.center);
		\draw [style=dik, in=-90, out=90] (75.center) to (78.center);
		\draw [style=dik, in=90, out=-90] (81.center) to (70.center);
		\draw [style=dik] (186.center) to (81.center);
		\draw [style=dik] (191.center)
			 to (192.center)
			 to (190.center)
			 to (189.center)
			 to cycle;
		\draw [style=grey dashed] (201.center) to (202.center);
		\draw [style=dik] (198.center)
			 to (200.center)
			 to (199.center)
			 to (197.center)
			 to cycle;
	\end{pgfonlayer}
\end{tikzpicture} \qquad b) \begin{tikzpicture}
	\begin{pgfonlayer}{nodelayer}
		\node [style=none] (0) at (1.75, -2) {};
		\node [style=none] (1) at (1.75, -3.25) {};
		\node [style=none] (2) at (3.75, -3.25) {};
		\node [style=none] (3) at (5.25, -3.25) {};
		\node [style=none] (4) at (7.25, -3.25) {};
		\node [style=none] (5) at (7.25, -2) {};
		\node [style=none] (6) at (3.75, -5.25) {};
		\node [style=none] (7) at (1.75, -5.25) {};
		\node [style=none] (8) at (1.75, -6.5) {};
		\node [style=none] (9) at (5.25, -5.25) {};
		\node [style=none] (10) at (7.25, -5.25) {};
		\node [style=none] (11) at (7.25, -6.5) {};
		\node [style=none] (12) at (3, -3.475) {};
		\node [style=none] (13) at (2.675, -5.05) {};
		\node [style=none] (15) at (6, -5) {};
		\node [style=none] (16) at (3, -3.25) {};
		\node [style=none] (17) at (2.675, -5.25) {};
		\node [style=none] (18) at (6, -5.25) {};
		\node [style=none] (20) at (5.25, -2) {};
		\node [style=none] (21) at (5.25, -0.75) {};
		\node [style=none] (22) at (3.65, -2) {};
		\node [style=none] (23) at (3.65, -0.75) {};
		\node [style=none] (24) at (3.75, -6.5) {};
		\node [style=none] (25) at (5.25, -6.5) {};
		\node [style=none] (26) at (5.25, -7.5) {};
		\node [style=none] (27) at (5.875, -1.5) {$T_A$};
		\node [style=none] (28) at (4.275, -1.5) {$C_A$};
		\node [style=none] (29) at (4.425, -6.975) {$C_A$};
		\node [style=none] (30) at (5.875, -6.975) {$T_A$};
		\node [style=none] (31) at (4.5, -2.7) {\small$\textsc{switch}$};
		\node [style=none] (44) at (2.525, -3.475) {};
		\node [style=none] (45) at (3.475, -3.475) {};
		\node [style=none] (46) at (3, -4.175) {};
		\node [style=none] (50) at (3, -3.75) {\scriptsize $1$};
		\node [style=none] (58) at (5.25, -0.75) {};
		\node [style=none] (59) at (4.725, -0.75) {};
		\node [style=none] (60) at (5.775, -0.75) {};
		\node [style=none] (61) at (4.9, -0.5) {};
		\node [style=none] (62) at (5.625, -0.5) {};
		\node [style=none] (63) at (5.1, -0.25) {};
		\node [style=none] (64) at (5.4, -0.25) {};
		\node [style=none] (65) at (4.5, -7.5) {};
		\node [style=none] (66) at (6, -7.5) {};
		\node [style=none] (67) at (5.25, -8.5) {};
		\node [style=none] (68) at (5.25, -7.9) {0};
		\node [style=none] (77) at (5.25, -10.25) {};
		\node [style=none] (78) at (5.25, -10.25) {};
		\node [style=none] (79) at (4.5, -10.25) {};
		\node [style=none] (80) at (8, -10.25) {};
		\node [style=none] (81) at (6.25, -12) {};
		\node [style=none] (82) at (6.25, -10.25) {};
		\node [style=none] (83) at (7.25, -10.25) {};
		\node [style=none] (84) at (6.25, -10.825) {\textsc{ghz}};
		\node [style=none] (85) at (7, -8.75) {};
		\node [style=none] (86) at (8, -8.75) {};
		\node [style=none] (89) at (9.75, -5.25) {\large$=$};
		\node [style=none] (119) at (12.9, -3.825) {$C_A$};
		\node [style=none] (149) at (13, -7.25) {};
		\node [style=none] (150) at (13, -7.25) {};
		\node [style=none] (151) at (12.25, -7.25) {};
		\node [style=none] (152) at (15.75, -7.25) {};
		\node [style=none] (153) at (14, -9) {};
		\node [style=none] (154) at (14, -7.25) {};
		\node [style=none] (155) at (15, -7.25) {};
		\node [style=none] (156) at (14, -7.825) {\textsc{ghz}};
		\node [style=none] (157) at (14.75, -5.75) {};
		\node [style=none] (158) at (15.75, -5.75) {};
		\node [style=none] (161) at (12.25, -5) {};
		\node [style=none] (162) at (12.25, -3.25) {};
		\node [style=none] (163) at (3.65, -0.75) {};
		\node [style=none] (164) at (3.65, -0.75) {};
		\node [style=none] (165) at (3.125, -0.75) {};
		\node [style=none] (166) at (4.175, -0.75) {};
		\node [style=none] (167) at (3.3, -0.5) {};
		\node [style=none] (168) at (4.025, -0.5) {};
		\node [style=none] (169) at (3.5, -0.25) {};
		\node [style=none] (170) at (3.8, -0.25) {};
		\node [style=none] (171) at (12.25, -3.25) {};
		\node [style=none] (178) at (6.025, -5) {};
		\node [style=none] (179) at (6.025, -5) {};
		\node [style=none] (180) at (6.025, -5) {};
		\node [style=none] (181) at (5.5, -5) {};
		\node [style=none] (182) at (6.55, -5) {};
		\node [style=none] (183) at (5.7, -4.75) {};
		\node [style=none] (184) at (6.375, -4.75) {};
		\node [style=none] (185) at (5.875, -4.5) {};
		\node [style=none] (186) at (6.175, -4.5) {};
		\node [style=none] (193) at (11.2, -3.25) {};
		\node [style=none] (194) at (13.3, -3.25) {};
		\node [style=none] (195) at (11.2, -2.5) {};
		\node [style=none] (196) at (13.3, -2.5) {};
		\node [style=none] (197) at (12.275, -2.5) {};
		\node [style=none] (198) at (12.275, -1.75) {};
		\node [style=grijs] (199) at (11.875, -1.75) {$a_1$};
		\node [style=none] (200) at (12.25, -2.875) {\scriptsize $X$ mmt};
		\node [style=none] (203) at (1.75, -5.05) {};
		\node [style=none] (204) at (3.575, -5.05) {};
		\node [style=none] (205) at (1.75, -4.3) {};
		\node [style=none] (206) at (3.575, -4.3) {};
		\node [style=none] (207) at (2.125, -4.3) {};
		\node [style=none] (208) at (1.3, -3) {};
		\node [style=grijs] (209) at (0.9, -2.5) {$a_1$};
		\node [style=none] (210) at (2.65, -4.675) {\scriptsize $Z$ mmt};
		\node [style=none] (211) at (1.3, -2.5) {};
		\node [style=none] (212) at (6.475, -3.475) {};
		\node [style=none] (213) at (6.475, -3.25) {};
		\node [style=none] (214) at (6, -3.475) {};
		\node [style=none] (215) at (6.95, -3.475) {};
		\node [style=none] (216) at (6.475, -4.175) {};
		\node [style=none] (217) at (6.475, -3.75) {\scriptsize $1$};
	\end{pgfonlayer}
	\begin{pgfonlayer}{edgelayer}
		\draw [style=dik filled] (3.center)
			 to (4.center)
			 to (5.center)
			 to (0.center)
			 to (1.center)
			 to (2.center)
			 to (6.center)
			 to (7.center)
			 to (8.center)
			 to (11.center)
			 to (10.center)
			 to (9.center)
			 to cycle;
		\draw [style=dik] (16.center) to (12.center);
		\draw [style=dik] (13.center) to (17.center);
		\draw [style=dik] (15.center) to (18.center);
		\draw [style=dik] (21.center) to (20.center);
		\draw [style=dik] (22.center) to (23.center);
		\draw [style=dik, in=90, out=-90] (25.center) to (26.center);
		\draw [style=dik] (46.center)
			 to (45.center)
			 to (44.center)
			 to cycle;
		\draw [style=dik] (63.center) to (64.center);
		\draw [style=dik] (61.center) to (62.center);
		\draw [style=dik] (59.center) to (60.center);
		\draw [style=dik] (67.center)
			 to (65.center)
			 to (66.center)
			 to cycle;
		\draw [style=dik, in=90, out=-90, looseness=1.25] (24.center) to (77.center);
		\draw [style=dik] (80.center)
			 to (81.center)
			 to (79.center)
			 to cycle;
		\draw [style=dik, in=-90, out=90] (82.center) to (85.center);
		\draw [style=dik, in=-90, out=90] (83.center) to (86.center);
		\draw [style=dik] (153.center)
			 to (151.center)
			 to (152.center)
			 to cycle;
		\draw [style=dik, in=-90, out=90] (154.center) to (157.center);
		\draw [style=dik, in=-90, out=90] (155.center) to (158.center);
		\draw [style=dik, in=90, out=-90] (161.center) to (150.center);
		\draw [style=dik] (162.center) to (161.center);
		\draw [style=dik] (169.center) to (170.center);
		\draw [style=dik] (167.center) to (168.center);
		\draw [style=dik] (165.center) to (166.center);
		\draw [style=dik] (185.center) to (186.center);
		\draw [style=dik] (183.center) to (184.center);
		\draw [style=dik] (181.center) to (182.center);
		\draw [style=grey dashed] (197.center) to (198.center);
		\draw [style=dik] (194.center)
			 to (196.center)
			 to (195.center)
			 to (193.center)
			 to cycle;
		\draw [style=grey dashed, in=-90, out=90] (207.center) to (208.center);
		\draw [style=dik] (203.center)
			 to (205.center)
			 to (206.center)
			 to (204.center)
			 to cycle;
		\draw [style=grey dashed] (208.center) to (211.center);
		\draw [style=dik] (213.center) to (212.center);
		\draw [style=dik] (216.center)
			 to (215.center)
			 to (214.center)
			 to cycle;
	\end{pgfonlayer}
\end{tikzpicture}}
    \caption{a) If $x_1=x_2=0$, then $\pA_1$ and $\pA_2$ perform identity operations, so that the quantum switch itself also reduces to the identity channel. b) If $x_1=x_2=1$, they both measure the target system in the computational ($Z$) basis and afterwards prepare the state $\ket1_{T_A}$. The $a_1$ outcome then has the same probability and postselected state as an $X$ basis measurement of the input control qubit. (Note that the switch is controlled in the $X$ basis.)}
    \label{fig:rmk-1-and-2}
\end{figure*}

Observations~\ref{fct:x-are-0} and~\ref{fct:x-are-1}, together with the fact that the control qubits are entangled in the GHZ state, tell us that if we make appropriate changes to the variables featuring in~\eqref{eq:ghz-cors} then the resulting implications will hold in our new possibility distribution~$q$:
\begin{equation}
    \label{eq:switchy-ghz-cors}
    \begin{split}
        \vec x=\mathbf1,\vec y=\mathbf1,\vec z=\mathbf1 \impl{q} a_1\oplus b_1 \oplus c_1 = 0; \\
        \vec x=\mathbf1,\vec y=\mathbf0,\vec z=\mathbf0 \impl{q} a_1\oplus b_3 \oplus c_3 = 1; \\
        \vec x=\mathbf0,\vec y=\mathbf1,\vec z=\mathbf0 \impl{q} a_3\oplus b_1 \oplus c_3 = 1; \\
        \vec x=\mathbf0,\vec y=\mathbf0,\vec z=\mathbf1 \impl{q} a_3\oplus b_3 \oplus c_1 = 1
    \end{split}
\end{equation}
(here $\mathbf0\coloneqq (0,0)$ and $\mathbf1\coloneqq (1,1)$).
(It is worth noting that these implications are nothing more than properties of the possibility distribution $q$ and can therefore also be derived directly from the Born rule and definition of the quantum switch. Observations~\ref{fct:x-are-0} and~\ref{fct:x-are-1} are only given to provide intuition for why~\eqref{eq:switchy-ghz-cors} should hold, given that we already know the GHZ correlations of~\eqref{eq:ghz-cors}.)

The implications in~\eqref{eq:switchy-ghz-cors}, like those in~\eqref{eq:ghz-cors}, do not immediately lead to a contradiction, because each equation on the right refers to outputs for a different set of inputs.
As in the case of GHZ nonlocality, a contradiction only arises when some of the outcome variables can be replaced by (hidden) variables that are independent of these measurement choices. The next sections define the assumptions of \emph{definite causal order} and \emph{relativistic causality} and show that they imply the existence of such variables---and, therefore, that the data $q(\vec a \vec b\vec c \vec x \vec y \vec z)$ observed in the quantum switch scenario just described are incompatible with the conjunction of those assumptions. This is formally stated in Theorem~\ref{thm:possibilistic}.

    \section{Assumption 1: Definite causal order}
    \label{sec:definite-causal-order}
    In the following we will want to probe causal relations on the basis of (possibilistic) correlations. An important assumption enabling this is that the input variables $x_1,x_2,y_1,y_2,z_1,z_2$ of the agents' interventions are freely chosen. In particular, this means that they are uncorrelated with variables outside their causal future, motivating some of the mathematical conditions defined below. It also motivates the assumption that all inputs are jointly possible, which will be needed later:
\begin{equation}
    \label{eq:all-inputs-possible}
    \forall \vec x,\vec y,\vec z\  q(\vec x\vec y\vec z) = 1.
\end{equation}

\begin{definition}
    A \defn{hidden causal order model} for a possibility distribution $r(a_1a_2x_1x_2)$ is a possibility distribution $r(a_1a_2x_1x_2\la)$, where $\la\in\{0,1\}$, such that
    \begin{subequations}
        \begin{align}
        {\sum}_\la r(a_1 a_2 x_1 x_2 \la) &= r(a_1 a_2 x_1 x_2), \\
            \la &\indep_r x_1x_2, \label{eq:hco-def-fc} \\
            a_1 &\indep_r x_2 \mid \la=0, \label{eq:hco-def-a1x2} \\
            \text{and\quad} a_2 &\indep_r x_1 \mid \la=1. \label{eq:hco-def-a2x1}
        \end{align}
    \end{subequations}
\end{definition}

The interpretation of this hidden variable $\la$ is that it determines the causal order between interventions performed by parties $\pA_1$ and $\pA_2$ with inputs $x_1,x_2$ and outputs $a_1,a_2$, respectively: the value $\la=0$ indicates that $\pA_1$ acts outside the causal future of $\pA_2$, so that the possibility of outcome $a_1$ is independent of the freely chosen input $x_2$.
The existence of a hidden causal order model is the possibilistic analogue of being a \defn{causal} probability distribution~\cite{Bran+15} (see also Appendix~\ref{app:independence-notation}).

\begin{assumption}[Definite causal order]
    \label{ass:dco}
    There exist three variables $\la_A,\la_B,\la_C\in\{0,1\}$ and a joint possibility distribution $q(\vec a\vec b\vec c\vec x\vec y\vec z\la_A\la_B\la_C)$ satisfying
    \begin{equation}
        \sum_{\la_A,\la_B,\la_C}q(\vec a\vec b\vec c\vec x\vec y\vec z\la_A\la_B\la_C) = q(\vec a\vec b\vec c\vec x\vec y\vec z),
    \end{equation}
    such that $q(a_1a_2x_1x_2\la_A)$ is a hidden causal order model for $q(a_1a_2x_1x_2)$, $q(b_1b_2y_1y_2\la_B)$ for $q(b_1b_2y_1y_2)$, and $q(c_1c_2z_1z_2\la_C)$ for $q(c_1c_2z_1z_2)$.%
    \footnote{Note that the assumptions discussed here are mathematically differently organised when compared to Ref.~\cite{vdLBC23}. The appropriate analogue of Assumption~\ref{ass:dco} here follows from the conjunction of the assumptions called Definite Causal Order and Free Interventions in~\cite{vdLBC23}. Likewise, the analogue of Assumption~\ref{ass:loc} below follows from Relativistic Causality and Free Interventions in~\cite{vdLBC23}.}
\end{assumption}

In the GHZ proof, the hidden variables were (when assumed to exist) directly revealed by outcomes of measurements, allowing us to replace the output variables $a,b,c$ in~\eqref{eq:ghz-cors} by hidden variables to obtain~\eqref{eq:ghz-hv-cors}.
This is not necessarily the case for hidden variables determining causal orders.
However, as we already noted after Observation~\ref{fct:x-are-1}, our choice of initial target state and operations performed inside the switch are such that, for $x_1=x_2=1$, each value of $a_1$ is compatible with only one of the orderings of the operations on the target system.
The following Lemma elevates this observation to the device-independent level.
\begin{lemma}
    \label{lma}
    Suppose that the possibility distribution $r(a_1a_2x_1x_2)$, where $a_1,a_2,x_1,x_2\in\{0,1\}$, satisfies
    \begin{subequations}
        \label{eq:lma-data}
        \begin{align}
            r(a_1=1,x_2=0) &= 0, \label{eq:lma-data-1} \\
            r(a_2=1,x_1=0) &= 0, \label{eq:lma-data-2} \\
            \text{and\quad} r(a_1=a_2=0,x_1=x_2=1) &= 0, \label{eq:lma-data-3}
        \end{align}
    \end{subequations}
    then any hidden causal order model $r(a_1a_2x_1x_2\la)$ for $r(a_1a_2x_1x_2)$ satisfies
    \begin{equation}
        x_1=x_2=1 \Longrightarrow_r a_1=\la.
    \end{equation}
\end{lemma}

\begin{proof}
    Our first claim is that
    \begin{equation}
        \label{eq:lma-prf-pt1}
        a_1=1\Longrightarrow_r \la=1.
    \end{equation}
    Indeed, from Equation~\eqref{eq:lma-data-1},
    \begin{equation}
        \begin{split}
            0 &= r(a_1=1, x_2=0) \\
            &\geq r(a_1=1,x_2=0,\la=0) \\
            &\stackrel{\eqref{eq:hco-def-a1x2}}{=} r(a_1=1,\la=0)r(x_2=0,\la=0) \\
            &\stackrel{\eqref{eq:hco-def-fc}}{=} r(a_1=1,\la=0)r(x_2=0)r(\la=0) \\
            &\stackrel{\eqref{eq:all-inputs-possible}}{=} r(a_1=1,\la=0)r(\la=0) \\
            &= r(a_1=1,\la=0). \\
        \end{split}
    \end{equation}
    Moreover, we have
    \begin{equation}
        \label{eq:lma-prf-pt2}
        x_1=x_2=1,a_1=0 \ \Longrightarrow_r\ a_2=1 \ \Longrightarrow_r\ \la=0;
    \end{equation}
    here the first implication is Equation~\eqref{eq:lma-data-3}, while the second follows from an argument analogous to that for~\eqref{eq:lma-prf-pt1} but using~\eqref{eq:lma-data-2}.~\eqref{eq:lma-prf-pt1} and~\eqref{eq:lma-prf-pt2} together imply the result.
\end{proof}

The marginals $q(a_1 a_2 x_1 x_2), q(b_1 b_2 y_1 y_2)$ and $q(c_1 c_2 z_1 z_2)$ for our three switches satisfy the conditions of Lemma~\ref{lma}. Therefore we can, under Assumption~\ref{ass:dco}, replace some of the observed variables in~\eqref{eq:switchy-ghz-cors} by hidden variables:
\begin{equation}
    \label{eq:switchy-ghz-hv-cors}
    \begin{split}
        \vec x=\mathbf1,\vec y=\mathbf1,\vec z=\mathbf1 & \impl{q} \la_A\oplus \la_B \oplus \la_C = 0; \\
        \vec x=\mathbf1,\vec y=\mathbf0,\vec z=\mathbf0 & \impl{q} \la_A\oplus b_3 \oplus c_3 = 1; \\
        \vec x=\mathbf0,\vec y=\mathbf1,\vec z=\mathbf0 & \impl{q} a_3\oplus \la_B \oplus c_3 = 1; \\
        \vec x=\mathbf0,\vec y=\mathbf0,\vec z=\mathbf1 & \impl{q} a_3\oplus b_3 \oplus \la_C = 1.
    \end{split}
\end{equation}

To derive a contradiction from this set of implications, we need our final assumption.

    \section{Assumption 2: Relativistic causality}
    \label{sec:locality}
    Recall that the three switches are implemented at spacelike separation. Moreover, the hidden variables $\la_{A,B,C}$ determine the causal order between the parties in their respective switch and must therefore take values in the causal past of those parties (see Figure~\ref{fig:switchy-ghz}).
With this in mind, the following conditions are motivated by the principle of relativistic causality, which requires the causal order on the variables to be compatible with the provided lightcone structure---meaning in particular that freely chosen inputs cannot influence the possibility of events outside their future lightcone.
\begin{assumption}[Relativistic causality]
    \label{ass:loc}
    The possibility distribution $q(\vec a\vec b\vec c\vec x\vec y\vec z\la_A\la_B\la_C)$ of Assumption~\ref{ass:dco} satisfies
    \begin{subequations}
        \label{eq:locality}
        \begin{align}
            \la_A\la_B\la_C &\indep_q \vec x\vec y\vec z; \label{eq:locality-1}\\
            \la_A b_3 c_3 &\indep_q \vec x \mid \vec y\vec z; \label{eq:locality-2}\\
            a_3 \la_B c_3 &\indep_q \vec y \mid \vec x\vec z; \label{eq:locality-3}\\
            a_3 b_3 \la_C &\indep_q \vec z \mid \vec x\vec y. \label{eq:locality-4}
        \end{align}
    \end{subequations}
\end{assumption}
Equation~\eqref{eq:locality-2}, for instance, is motivated by the existence of the spacelike hypersurface depicted in orange in Figure~\ref{fig:switchy-ghz}.
Note that this Equation implies conditional independences like $b_3 c_3 \indep_q \vec x\mid \la_A \vec y \vec z$ and $\la_A \indep_q \vec x \mid \vec y \vec z$, which are akin to the probabilistic notions that are often referred to as \emph{parameter independence} and \emph{measurement independence}, respectively (see also Appendix~\ref{app:independence-notation}).
Here, we consider both these notions to be motivated by the principle of relativistic causality, having assumed free choice.
The four conditions~\eqref{eq:locality-1}--\eqref{eq:locality-4} can be directly applied to~\eqref{eq:switchy-ghz-hv-cors} to yield the four equations%
\begin{equation}
    \label{eq:switchy-ghz-hv-loc-cors}
    \begin{alignedat}{2}
        \vec x=\mathbf0,\vec y=\mathbf0,\vec z=\mathbf0 \impl{q} && \la_A\oplus \la_B \oplus \la_C &= 0 \\
        \vec x=\mathbf0,\vec y=\mathbf0,\vec z=\mathbf0 \impl{q} && \la_A\oplus b_3 \oplus c_3 &= 1 \\
        \vec x=\mathbf0,\vec y=\mathbf0,\vec z=\mathbf0 \impl{q} && a_3\oplus \la_B \oplus c_3 &= 1 \\
        \vec x=\mathbf0,\vec y=\mathbf0,\vec z=\mathbf0 \impl{q} && a_3\oplus b_3 \oplus \la_C &= 1 \quad \oplus \\
        \noalign{\smallskip} \hline \noalign{\smallskip}
        \vec x=\mathbf0,\vec y=\mathbf0,\vec z=\mathbf0 \impl{q} && 0 &=1.
    \end{alignedat}
\end{equation}
Together with the fact that the combination $\vec x=\vec y=\vec z=\mathbf0$ is possible (by Eq.~\eqref{eq:all-inputs-possible}), this yields our desired contradiction.
Summarising, we have proven the following.
\begin{theorem}\label{thm:possibilistic}
    No possibility distribution $q(\vec a \vec b \vec c \vec x \vec y \vec z)$ that satisfies Equations~\eqref{eq:switchy-ghz-cors},~\eqref{eq:all-inputs-possible},~\eqref{eq:lma-data} and similar conditions for $B$ and $C$ can also satisfy the conjunction of Assumptions~\ref{ass:dco} (definite causal order) and~\ref{ass:loc} (relativistic causality).
    In particular, the distribution $q(\vec a \vec b \vec c \vec x \vec y \vec z)$ arising in the quantum switch scenario depicted in Figure~\ref{fig:switchy-ghz} is incompatible with these assumptions.
\end{theorem}
\hfill

    \section{Back to probabilities: a causal Mermin inequality}
    Although quantum theory makes predictions about the (im)possibility of events in ideal measurement scenarios, experimental data are inevitably subject to uncertainties and are thus bound to be probabilistic.
Impossibilities like those of Eqs.~\eqref{eq:switchy-ghz-cors} and~\eqref{eq:lma-data} can therefore not be verified with certainty, leaving the possibilistic argument of the preceding sections unamenable to experiment.
However, when the probabilistic data approximately satisfy the required impossibilities, a refutation of (probabilistic versions of) the assumptions is still possible via the violation of the following inequality, closely analogous to Mermin's inequality~\cite{Mer90a} for Bell nonlocality.
The derivation of this inequality relies on the probabilistic conditions~\eqref{eq:probabilistic-conditions} given in Appendix~\ref{app:mermin-ineq}, which are analogous to (but stronger than) the possibilistic conditions expressed in Equation~\eqref{eq:all-inputs-possible} and Assumptions~\ref{ass:dco} and~\ref{ass:loc}, and are likewise physically motivated by the principles of free interventions, definite causal order, and relativistic causality.


\begin{theorem}\label{thm:causal-mermin}
    Any probability distribution $P(\vec a \vec b \vec c \vec x\vec y \vec z \la_A \la_B \la_C)$ that fulfills the conditions in~\eqref{eq:probabilistic-conditions} satisfies the inequality
    \begin{widetext}
        \begin{alignat}{5}  
            & P(a_1 \oplus b_1\oplus c_1 = 0 \mid \vec x=\mathbf1,\vec y=\mathbf1,\vec z = \mathbf1)
            + P(a_1 \oplus b_3\oplus c_3 = 1 \mid \vec x=\mathbf1,\vec y=\mathbf0,\vec z = \mathbf0) \span \span \span \span \span \nonumber \\
            +{} & P(a_3 \oplus b_1\oplus c_3 = 1 \mid \vec x=\mathbf0,\vec y=\mathbf1,\vec z = \mathbf0)
            + P(a_3 \oplus b_3\oplus c_1 = 1 \mid \vec x=\mathbf0,\vec y=\mathbf0,\vec z = \mathbf1)  \span \span \span \span \span \nonumber \\
            -{} 2\big[\ &P(a_1=1 \mid x_1x_2=10) &&+{} P(a_2=1 \mid x_1 x_2 = 01) &&+{} P(a_1 a_2=00 \mid x_1x_2=11) &&& \label{eq:causal-mermin-ineq} \\
            +{} &P(b_1=1 \mid y_1y_2=10) &&+{} P(b_2=1 \mid y_1 y_2 = 01) &&+{} P(b_1 b_2=00 \mid y_1y_2=11) \nonumber \\
            +{} &P(c_1=1 \mid z_1z_2=10) &&+{} P(c_2=1 \mid z_1 z_2 = 01) &&+{} P(c_1 c_2=00 \mid z_1z_2=11)  \big] &&\leq 3. \nonumber
        \end{alignat}
    \end{widetext}
    \begin{proof}
        See Appendix~\ref{app:mermin-ineq}.
    \end{proof}
\end{theorem}

A violation of this inequality to the algebraic bound of~4 requires precisely the possibilistic properties expressed in Eqs.~\eqref{eq:switchy-ghz-cors} and~\eqref{eq:lma-data}, and is therefore attained by precisely those probabilistic data that also admit a possibilistic derivation of a contradiction with the conjunction of~\eqref{eq:all-inputs-possible} and Assumptions~\ref{ass:dco} and~\ref{ass:loc} as in Sections~\ref{sec:switch-data}--\ref{sec:locality}---including the data $Q$ predicted for the quantum switch scenario of Figure~\ref{fig:switchy-ghz}.

It is worth noting that, as Mermin did with his three-party GHZ scenario~\cite{Mer90a}, the three-switch scenario considered here can be suitably generalised to one with $N>3$ spacelike-separated switches. Definite causal order and relativistic causality then imply an inequality with bound $2^{\lfloor(3N-4)/2\rfloor}$, which is violated by the switch scenario to the algebraic maximum of $2^{2N-3}$. The proportional violation of the inequality thus grows exponentially with $N$.

    \section{Discussion}
    When it was shown in~\cite{vdLBC23, GP23} that the quantum switch exhibits indefinite causal order that can be device-independently tested, this meant that it predicts probabilities that 
cannot be recovered from a joint \emph{probability} distribution with a variable that constrains the causal order on \emph{every} run of the experiment (in such a way that the order is always in accordance with relativity theory).
Here, we have strengthened that result in two ways.

First, the argument in Sections~\ref{sec:switch-data}--\ref{sec:locality} relies only on the calculus of possibilities, rather than probabilities, so it eliminates the assumption that probability theory is applicable on the hidden variable level.
This assumption is nontrivial: even if the variables $\la_{A,B,C}$ are postulated to take values on every run of the experiment, their frequencies might not converge to probabilities satisfying the probability axioms.
And while there are empirical grounds for assuming such convergence in the case of observable quantities, these do not necessarily apply to the causal order variables, as those might be unobservable in principle.
In addition, the possibilistic independence assumptions involved in our argument of Sections~\ref{sec:switch-data}--\ref{sec:locality} are much weaker than the probabilistic ones required in~\cite{vdLBC23,GP23},
with relativistic causality conditions like~\eqref{eq:locality} merely requiring that a choice of input cannot influence the \emph{possibility}, rather than probability, of events that are spacelike-separated or lie to its past.

The maximal violation of the causal Mermin inequality~\eqref{eq:causal-mermin-ineq}, meanwhile, strengthens the results of~\cite{vdLBC23,GP23} by showing that the quantum switch (or any other data that reach the algebraic maximum) is not even compatible with a (probabilistic) hidden variable model that specifies a causal order \emph{sometimes}: any nonzero fraction of runs in which all operations are causally ordered and no superluminal influences occur leads to a value strictly below the algebraic bound.\footnote{In the terminology of~\cite{GP23}, this means that the data have a causally separable fraction of 0\% over the relevant spaces of input histories.}
In this sense the device-independent indefinite causal order in the quantum switch can be said to be maximal.
It is shown in Appendix~\ref{app:chained-ineqs} that this result can in fact also be achieved with a single quantum switch, by using chained CHSH inequalities~\cite{BC90} and taking an appropriate limit.
This is analogous to the observation in the context of Bell nonlocality that quantum correlations in the GHZ and chained CHSH scenario are maximally Bell-nonlocal: they have no Bell-local fraction~\cite{BKP06}.

An important way in which this work, as well as Refs.~\cite{vdLBC23,GP23,HO21}, differ from most other literature on indefinite causal order is that they consider constraints on the allowed causal orders.
Here and in~\cite{vdLBC23}, these constraints were motivated by the principle of relativistic causality, but it is worth noting that they could also be motivated differently.
Ref.~\cite{GP23} propose a general framework for constrained indefinite causality.
See Appendix~\ref{app:three-switches} for a comparison of the present work with a result from~\cite{GP23}.

Our proofs here are much like those of Bell nonlocality by GHZ and Mermin, except that the existence of hidden variables predetermining measurement outcomes is derived from the definite causal order assumption together with properties of the observed data, rather than directly assumed or concluded from an EPR argument.
This connection between Bell nonlocality and indefinite causal order, initially found in Ref.~\cite{vdLBC23} in the context of a CHSH inequality, thus extends to possibilistic arguments as well.
While this article makes this explicit in the GHZ scenario, a Hardy-type argument~\cite{Har93} can also be constructed.
This work and Ref.~\cite{vdLBC23} can, in light of this connection with Bell nonlocality, be compared to recent works on extended Wigner's friend scenarios (see for instance Refs.~\cite{Bong+20,HC22,OVB23}), in which not a definite causal order, but the observation of an agent is what predetermines a measurement outcome.%
\footnote{
    Our notion of relativistic causality is analogous to what is called Locality or Local Agency in~\cite{Bong+20,HC22}, and to what is known as parameter independence (and measurement independence) in the context of Bell's theorem~\cite{Shi86}.
    It is strictly weaker than the kind of locality already ruled out by Bell~\cite{Bell64}, GHZ~\cite{GHSZ90} and Mermin~\cite{Mer90}.
    On its own it is consistent with quantum physics; only in the presence of another assumption, like definite causal order (Theorems~\ref{thm:possibilistic} and~\ref{thm:causal-mermin} and~\cite{vdLBC23}), absoluteness of an agent's observation~\cite{Bong+20, HC22}, determinism (Section~\ref{sec:vanilla-ghz} and~\cite{Bell64,GHZ89}), or alternatively, outcome independence~\cite{Shi86,Mer90a}, can contradictions be derived.
}
It will be worth investigating what other device-independent tests can be devised using this method, going beyond causal order and Wigner's friend scenarios.

Another important direction for future work is to understand the physical implications of experimental violations of the inequalities presented here and in Ref.~\cite{vdLBC23}.
The notion of causal order, and whether it is definite or not, naturally depends on what is being ordered, i.e.\ on what are taken as the relata of the causal relations.
In the device-independent treatment, it is customary to assume that each relatum covers both the choosing of an input and the generation of an output.
This is also the approach taken here and in Refs.~\cite{vdLBC23,GP23}.
This tradition goes back to Bell's theorem, in which such input-output pairs can be confined to a relatively small spacetime region.
Current implementations of the quantum switch however necessitate that the inputs of the operations inside the switch are chosen early, and the outcomes are only measured after execution of the switch (e.g.~\cite{Rub+17,Cao+23}).
Consequently, (classical) relativity theory does not rule out two-way communication between the regions occupied by the input-output pairs.
Is indefinite causal order with respect to these relata then still an interesting physical phenomenon, or is another choice of relata more natural?
Answering these questions might require a more in-depth analysis which investigates whether such two-way communication is indeed facilitated by the used experimental setups.

    \section*{Acknowledgements}
    We would like to thank Jonathan Barrett, Giulio Chiribella, Marc-Olivier Renou and Nicola Pinzani for insightful discussions.
    This research was funded in part by the Engineering and Physical Sciences Research Council (EPSRC),
    and was supported by the John Templeton Foundation through the ID\# 62312 grant, as part of the \href{https://www.templeton.org/grant/the-quantum-information-structure-of-spacetime-qiss-second-phase}{‘The Quantum Information Structure of Spacetime’} Project (QISS). The opinions expressed in this publication are those of the authors and do not necessarily reflect the views of the John Templeton Foundation. For the purpose of Open Access, the authors have applied a CC BY public copyright licence to any Author Accepted Manuscript (AAM) version arising from this submission.

    \bibliographystyle{quantum}
    \bibliography{zotero}

    \onecolumn
    \clearpage
    \newgeometry{left=3cm,right=3cm,top=3cm,bottom=3cm}
    \appendix

    \section{The quantum switch correlations and formalisation of Observations~\ref{fct:x-are-0} and~\ref{fct:x-are-1}}
    \label{app:switch-data}
    Here we describe in more detail the outcome probabilities (and thus possibilities) arising in the three-switch scenario described in the main text.
We then state and prove formal versions of Observations~\ref{fct:x-are-0} and~\ref{fct:x-are-1}.

The interventions of each of the nine agents $\pA_{1,2,3},\pB_{1,2,3},\pC_{1,2,3}$ can, for each value of their classical input variable, be described by a quantum instrument, i.e.\ a quantum operation with a classical outcome.
A quantum instrument, with input and output systems described by Hilbert spaces $\H$ and $\K$ and with finite classical outcome set $E$, is a collection of completely positive (CP) maps $\{\E^e : \cB(\H) \to \cB(\K)\}_{e\in E}$ that sum to a trace-preserving map $\sum_{e\in E}\E^e$. (Here $\cB(\H)$ and $\cB(\K)$ are the spaces of bounded linear operators on $\H$ and $\K$.) Each map $\E^e$ represents the quantum operation that ends up being performed in case outcome $e$ is observed.

For $i=1,2$, the instrument performed by $\pA_i$ on the target system $T_A$ for a fixed value of $x_i$, described intuitively in the main text, is defined by the CP maps $\cA_i^{a_i|x_i}:\cB(\H_{T_A})\to\cB(\H_{T_A})$ given by
\begin{subequations}
    \begin{align}
        \cA_i^{a_i|x_i=0} &= \delta_{a_i=0} \cdot \id_{T_A} \\
        \text{and\quad}\cA_i^{a_i|x_i=1}(\rho_{T_A}) &= \op{1}{a_i}\!\rho_{T_A}\!\op{a_i}{1}, \label{eq:A1x1-instr}
    \end{align}
\end{subequations}
where $\id_{T_A}$ is the identity operation on $\cB(\H_{T_A})$ and $\delta_{a_i=0}$ is $1$ if $a_i=0$ and 0 otherwise.
$\pA_3$ performs a projective measurement in the $Y$ basis, described by CP maps $\cA_3^{a_3} : \cB(\H_{C_A}) \to \cB(\mathbb C)$ with trivial output system:
\begin{equation}
    \cA_3^{a_3=0}(\rho_{C_A}) = \ev{\rho_{C_A}}{+i}; \qquad \cA_3^{a_3=1}(\rho_{C_A}) = \ev{\rho_{C_A}}{-i}.
\end{equation}

Together, when situated in one wing of the quantum switch setup considered in the main text, these instruments define the `joint' instrument $\{\cA^{\vec a|\vec x} : \cB(\H_{C_A})\to\cB(\H_{C_A})\}_{\vec a}$, with
\begin{equation}
    \label{eq:joint-instrument}
    \cA^{\vec a|\vec x}(\rho_{C_A}) \coloneqq \left((\cA_3^{a_3} \tns \Tr_{T_A}) \circ \switch(\cA_1^{a_1|x_1}, \cA_2^{a_2|x_2}) \right)(\rho_{C_A}\tns\dyad0_{T_A}).
\end{equation}
The joint instruments $\cB^{\vec b|\vec y}$ and $\cC^{\vec c|\vec z}$ are defined similarly.
The joint probabilities of the outcomes $\vec a, \vec b, \vec c$ given inputs $\vec x, \vec y, \vec z$ are then given by the distribution
\begin{equation}
    Q(\vec a \vec b \vec c \mid \vec x \vec y \vec z) \coloneqq \left(\cA^{\vec a|\vec x} \tns \cB^{\vec b|\vec y} \tns \cC^{\vec c|\vec z}\right)\left( \dyad{\GHZ}_{C_A C_B C_C} \right),
\end{equation}
which is defined diagrammatically~\cite{PQP} in Figure~\ref{fig:switchy-ghz}.
Finally, together with a distribution over the input variables $Q(\vec x\vec y\vec z)$, this gives a probability distribution $Q(\vec a \vec b \vec c \vec x \vec y \vec z)$, which in turn leads to the possibility distribution $q(\vec a \vec b \vec c \vec x \vec y \vec z)$ considered in the main text.

It is straightforward to verify that the marginals of this distribution satisfy the conditions~\eqref{eq:lma-data} of Lemma~\ref{lma}.
We can now also formalise and prove Observations~\ref{fct:x-are-0} and~\ref{fct:x-are-1} from the main text. This is done by Lemmas~\ref{lma:obs-1} and~\ref{lma:obs-2} below, respectively.
\begin{lemma}\label{lma:obs-1}
    We have
    \begin{equation}
        \cA^{\vec a|\vec x=\mathbf0}(\rho_{C_A}) = \delta_{a_1=a_2=0} \cdot \cA_3^{a_3}(\rho_{C_A}).
    \end{equation}

    \begin{proof}
        This follows from the fact that $\switch\left(\id_{T_A}, \id_{T_A}\right) = \id_{C_A T_A}$, which is evident from the definition of the quantum switch in Eq.~\eqref{eq:switch-defn}.
    \end{proof}
\end{lemma}

\begin{lemma}\label{lma:obs-2}
    We have
    \begin{equation}
        \sum_{a_2 a_3} \cA^{\vec a|\vec x=\mathbf1}(\rho_{C_A}) =
        \begin{cases}
            \ev{\rho_{C_A}}{+} & \text{ if } a_1=0; \\
            \ev{\rho_{C_A}}{-} & \text{ if } a_1=1.
        \end{cases}
    \end{equation}

    \begin{proof}
        We prove the fact for pure states $\rho_{C_A} = \dyad\psi_{C_A}$; by linearity this implies the result for mixed states too.
        From the definition of the quantum switch~\eqref{eq:switch-defn} and the quantum instruments~(\ref{eq:A1x1-instr},~\ref{eq:joint-instrument}), we have
        \begin{equation}
            \begin{split}
                \cA^{\vec a|\vec x=\mathbf1}(\dyad\psi_{C_A}) = \Big\| \bra{(-1)^{a_3}i}
                \Big( & \op{+}_{C_A} \tns \ket{1}\ip{a_2}{1}\bra{a_1}_{T_A} + \\
                & \op{-}_{C_A} \tns \ket{1}\ip{a_1}{1}\bra{a_2}_{T_A} \Big) \ket\psi_{C_A}\ket0_{T_A}
                \Big\|^2.
            \end{split}
        \end{equation}
        If $a_1=0$ then the second term (in which $\pA_1$ is after $\pA_2$) vanishes, giving
        \begin{equation}
            \cA^{\vec a|\vec x=\mathbf1}(\dyad\psi_{C_A}) = \left| \ip{(-1)^{a_3}i}{+}_{C_A} \ip{+}{\psi}_{C_A} \cdot\delta_{a_2=1} \right|^2
        \end{equation}
        so that
        \begin{equation}
            \sum_{a_2 a_3} \cA^{\vec a|\vec x=\mathbf1}(\dyad\psi_{C_A}) = \left|\ip{+}{\psi}\right|^2.
        \end{equation}
        The case for $a_1=1$ is similar.
    \end{proof}
\end{lemma}

Together with the known GHZ correlations~\eqref{eq:ghz-cors}, these facts imply the correlations in~\eqref{eq:switchy-ghz-cors}, used in our possibilistic argument, as well as the fact that the quantum switch setup violates, with the choice of instruments described here, the causal Mermin inequality~\eqref{eq:causal-mermin-ineq} to the algebraic bound of 4.

    \section{On conditional independences and possibility distributions}
    \label{app:independence-notation}
    Throughout this work we use the notation $\indep$ for conditional independence relations in both possibility and probability distributions.
Here we discuss in some more detail how this notation relates to that more commonly found in the literature on Bell nonlocality and causal inequalities (see e.g.\ Refs.~\cite{Bran+15,Bong+20,Shi86}).

Consider first the case for probability distributions, discussed explicitly only in Appendix~\ref{app:mermin-ineq}.
As is somewhat customary, we will abuse notation by suppressing the distinction between random variables and the values that they take,
so that $P(abc)$ sometimes refers to a probability, and sometimes to a probability distribution.
It will be convenient to use the convention that the conditional probability $P(ab|c)$ is $P(abc)/P(c)$ if $P(c) > 0$, and $0$ if $P(c) = 0$.
Conditional independence of $a$ and $b$ given $c$ can then be defined as
\begin{equation}\label{eq:probabilistic-conditional-independence}
    \quad a\indep_P b\mid c \quad :\!\iff \quad \forall a,b,c: P(ab\mid c) = P(a\mid c)P(b\mid c).
\end{equation}
Just in case $P(b|c) > 0$ for all $b,c$, this condition is equivalent to
\begin{equation}
    \forall a,b,c: P(a\mid bc) = P(a\mid c), \quad\text{ or, in other words, } \quad \forall a,b,b',c: P(a\mid bc) = P(a\mid b'c).
\end{equation}

As an example, the condition $a_1 \indep_P x_2 \mid \la_A=0$ in Eq.~\eqref{eq:probabilistic-dco-a} in Appendix~\ref{app:mermin-ineq}---which encodes part of the conditions of a definite causal order~\cite{Bran+15} and is analogous to the \emph{possibilistic} condition found in Eq.~\eqref{eq:hco-def-a1x2}---is equivalent to
\begin{equation}
    \forall a_1,x_2,x_2': P(a_1 \mid x_2, \la_A = 0) = P(a_1 \mid x_2', \la_A=0)
\end{equation}
in case $P(x_2 \mid \la_A = 0) > 0$ for all $x_2$.
Consider, as another example, the condition $\la_A b_3 c_3 \indep_P \vec x \mid \vec y \vec z$ from Eq.~\eqref{eq:probabilistic-rc}, analogous to the possibilistic condition stated in~\eqref{eq:locality-2}.
This implies the weaker condition $b_3 c_3 \indep_P \vec x \mid \la_A \vec y \vec z$, which can also be stated as
\begin{equation}
    \forall b_3,c_3,\la_A,\vec x, \vec x', \vec y, \vec z: P(b_3 c_3 \mid \la_A \vec x \vec y \vec z) = P(b_3 c_3 \mid \la_A \vec x' \vec y \vec z)
\end{equation}
so long as $P(\la_A \mid \vec y\vec z) > 0$ for all values of $\la_A,\vec y$, and $\vec z$.
This is the type of condition variably referred to as `parameter independence'~\cite{Shi86}, `locality'~\cite{Bong+20}, or `no signalling on the hidden-variable level' in Bell nonlocality and related literature.

\

Possibilities take only the values $0$ (impossible) and $1$ (possible), which satisfy, in particular, $1+1=1$ (see Section~\ref{sec:notation}).
For possibilities, assuming strict positivity is often undesirable; and while one may still define conditional possibilities similarly to the way done for probabilities, they are less meaningful.%
\footnote{This is one reason why, in the main text, we treat the setting variables $\vec x, \vec y, \vec z$ as part of the distribution $p(\vec a \vec b \vec c \vec x \vec y \vec z)$ rather than as conditional variables in a conditional distribution $p(\vec a \vec b \vec c \mid \vec x \vec y \vec z)$.}
One can still reason about conditional independence relations as follows.
First of all, note that a natural notion of \emph{un}conditional independence of variables $a,b$ in a possibility distribution $p(ab)$ is that
\begin{equation}
    \begin{split}
        a\indep_p b \quad:\!&\iff\quad \forall a,b: p(ab) = p(a)p(b) \\
        &\iff\quad \forall a,b: p(a) = p(b) = 1 \Rightarrow p(ab) = 1;
    \end{split}
\end{equation}
in other words, if $a$ and $b$ are individually possible, then they are jointly possible.
An appropriate notion of conditional independence of $a$ and $b$ given $c$ is, by extension, that joint possibility of $a$ and $c$ together with joint possibility of $b$ and $c$ imply joint possibility of $a,b$ and $c$ all together:
\begin{equation}
    \begin{split}
        a\indep_p b \mid c \quad:\!&\iff\quad \forall a,b,c: p(abc) = p(ac)p(bc) \\
        &\iff\quad \forall a,b,c: p(ac) = p(bc) = 1 \Rightarrow p(abc) = 1.
    \end{split}
\end{equation}

In the case of a possibility distribution $p(abc) = \pi(P(abc))$ arising from a probability distribution $P(abc)$ via the map $\pi$ defined in Section~\ref{sec:notation}, probabilistic independence implies possibilistic independence: $a\indep_P b \mid c \implies a\indep_{p} b \mid c$.
The latter condition, however, is strictly weaker.

    \section{Proof of Theorem~\ref{thm:causal-mermin}}
    \label{app:mermin-ineq}
    In the below, $\indep_P$ denotes \emph{probabilistic} conditional independence in a probability distribution $P$, as defined in Eq.~\eqref{eq:probabilistic-conditional-independence}.
The following is a restatement of Theorem~\ref{thm:causal-mermin}.

\begin{theorem*}
    \label{thm:causal-mermin-app}
    Let $P(a_1 a_2 a_3 b_1 b_2 b_3 c_1 c_2 c_3 x_1 x_2 y_1 y_2 z_1 z_2 \la_A \la_B \la_C) \eqqcolon P(\vec a \vec b \vec c \vec x \vec y \vec z \la_A \la_B \la_C)$ be a probability distribution on variables taking values in $\{0,1\}$ that satisfies the following conditions, physically motivated as in the main text by the assumptions of free interventions~\eqref{eq:probabilistic-all-inputs-possible}, definite causal order~\eqref{eq:probabilistic-dco-a}--\eqref{eq:probabilistic-dco-c}, and relativistic causality~\eqref{eq:probabilistic-rc}--\eqref{eq:probabilistic-rc-extra}:%
    \footnote{Note that unlike the other conditions in~\eqref{eq:probabilistic-conditions}, condition~\eqref{eq:probabilistic-rc-extra} does not have a counterpart in the main text.
    In brief, that is because it is not needed when assuming the perfect theoretical predictions of the quantum switch data, whilst in this appendix we wish to accommodate for noise or data that otherwise deviate from those of the quantum switch.
    The justification for~\eqref{eq:probabilistic-rc-extra}, just like~\eqref{eq:probabilistic-rc}, is the principle of relativistic causality.}
    \begin{subequations}
        \label{eq:probabilistic-conditions}
        \begin{align}
            \forall \vec x \vec y \vec z: P(\vec x \vec y \vec z) &> 0; \label{eq:probabilistic-all-inputs-possible} \\
            a_1\indep_P x_2 \mid \la_A=0;  \qquad & a_2\indep_P x_1 \mid \la_A=1; \label{eq:probabilistic-dco-a}\\
            b_1\indep_P y_2 \mid \la_B=0;  \qquad & b_2\indep_P y_1 \mid \la_B=1; \\
            c_1\indep_P z_2 \mid \la_C=0;  \qquad & c_2\indep_P z_1 \mid \la_C=1; \label{eq:probabilistic-dco-c}\\
            \la_A\la_B\la_C \indep_P \vec x\vec y\vec z; \qquad \la_A b_3 c_3 \indep_P \vec x \mid \vec y\vec z; \qquad & a_3 \la_B c_3 \indep_P \vec y \mid \vec x\vec z; \qquad a_3 b_3 \la_C \indep_P \vec z \mid \vec x\vec y; \label{eq:probabilistic-rc} \\
            a_1 a_2 \la_A \indep_P \vec y \vec z \mid \vec x; \qquad b_1 b_2 \la_B \indep_P \vec x &\vec z \mid \vec y; \qquad c_1 c_2 \la_C \indep_P \vec x \vec y \mid \vec z. \label{eq:probabilistic-rc-extra}
        \end{align}
    \end{subequations}
    Then $P$ satisfies the inequality
    \begin{alignat}{5}  
        & P(a_1 \oplus b_1\oplus c_1 = 0 \mid \vec x=\mathbf1,\vec y=\mathbf1,\vec z = \mathbf1)
        + P(a_1 \oplus b_3\oplus c_3 = 1 \mid \vec x=\mathbf1,\vec y=\mathbf0,\vec z = \mathbf0) \span \span \span \span \span \nonumber \\
        +{} & P(a_3 \oplus b_1\oplus c_3 = 1 \mid \vec x=\mathbf0,\vec y=\mathbf1,\vec z = \mathbf0)
        + P(a_3 \oplus b_3\oplus c_1 = 1 \mid \vec x=\mathbf0,\vec y=\mathbf0,\vec z = \mathbf1)  \span \span \span \span \span \nonumber \\
        -{} 2\big[\ &P(a_1=1 \mid x_1x_2=10) &&+{} P(a_2=1 \mid x_1 x_2 = 01) &&+{} P(a_1 a_2=00 \mid x_1x_2=11) &&& \label{eq:causal-mermin-ineq-app} \\
        +{} &P(b_1=1 \mid y_1y_2=10) &&+{} P(b_2=1 \mid y_1 y_2 = 01) &&+{} P(b_1 b_2=00 \mid y_1y_2=11) \nonumber \\
        +{} &P(c_1=1 \mid z_1z_2=10) &&+{} P(c_2=1 \mid z_1 z_2 = 01) &&+{} P(c_1 c_2=00 \mid z_1z_2=11)  \big] &&\leq 3. \nonumber
    \end{alignat}
\end{theorem*}
\begin{proof}
    The proof consists of two parts: we will first see that definite causal order~\eqref{eq:probabilistic-dco-a}--\eqref{eq:probabilistic-dco-c} implies some amount of determinism, quantified by the last three lines of the inequality, and then how that determinism, together with relativistic causality~\eqref{eq:probabilistic-rc}--\eqref{eq:probabilistic-rc-extra}, provides a Mermin-like bound on the first two lines of the inequality.
    (The first part parallels our going from~\eqref{eq:switchy-ghz-cors} to~\eqref{eq:switchy-ghz-hv-cors} in the main text, while the second parallels the transition to~\eqref{eq:switchy-ghz-hv-loc-cors} by relativistic causality.) 

    Denote the final three lines in the inequality by $\a,\b$ and $\ga$:
    \begin{subequations}
        \begin{align}
            \a[P] &\coloneqq P(a_1=1 \mid x_1 x_2=10) + P(a_2=1 \mid x_1 x_2=01) + P(a_1 a_2=00 \mid x_1 x_2=11); \\
            \b[P] &\coloneqq P(b_1=1 \mid y_1 y_2=10) + P(b_2=1 \mid y_1 y_2=01) + P(b_1 b_2=00 \mid y_1 y_2=11); \\
            \ga[P] &\coloneqq P(c_1=1 \mid z_1 z_2=10) + P(c_2=1 \mid z_1 z_2=01) + P(c_1 c_2=00 \mid z_1 z_2=11).
        \end{align}
    \end{subequations}
    (The probabilities are independent of the omitted input variables by~\eqref{eq:probabilistic-rc-extra}.)
    The following mirrors Lemma~\ref{lma} in the main text, but is probabilistic and robust to noise in the values of $\a[P],\b[P],$ and $\ga[P]$, which are exactly zero in the ideal quantum switch scenario described in the main text.

    \begin{claim*}
        We have
        \begin{subequations}
            \label{eq:a1-la-alpha-etc}
            \begin{align}
                P(a_1=\la_A \mid \vec x = \mathbf1) &\geq 1-\a[P]; \label{eq:a1-la-alpha}\\
                P(b_1=\la_B \mid \vec y = \mathbf1) &\geq 1-\b[P]; \label{eq:b1-la-beta} \\
                P(c_1=\la_C \mid \vec z = \mathbf1) &\geq 1-\ga[P]. \label{eq:c1-la-gamma}
            \end{align}
        \end{subequations}
        \begin{proof}
            In the below, the unnamed conditioned-upon variables are always $x_1$ and $x_2$.
            Using~\eqref{eq:probabilistic-dco-a}, we get
            \begin{align}
                P(a_1=1\mid 11, \la_A=0) &= P(a_1=1\mid 10, \la_A=0) \leq \a[P(\,\cdot\mid\la_A=0)] \text{\quad and } \\
                \begin{split}
                    P(a_1=0\mid 11, \la_A=1) &= P(a_1=0, a_2=0\mid 11, \la_A=1) + P(a_1=0,a_2=1\mid 11, \la_A=1) \\
                    &\leq P(a_1=0, a_2=0\mid 11,\la_A=1) + P(a_2=1\mid 11,\la_A=1)  \\
                    &= P(a_1=0, a_2=0\mid 11, \la_A=1) + P(a_2=1\mid 01, \la_A=1) \\
                    &\leq \a[P(\,\cdot\mid\la_A=1)].
                \end{split}
            \end{align}
            Now, using $\la_A\la_B\la_C \indep_P \vec x\vec y\vec z$ from~\eqref{eq:probabilistic-rc}, we can conclude
            \begin{align}
                P(a_1\neq \la_A\mid 11) &= P(\la_A=0)P(a_1=1\mid 11,\la_A=0) + P(\la_A=1)P(a_1=0\mid 11,\la_A=1) \\
                &\leq P(\la_A=0) \a[P(\,\cdot\mid\la_A=0)] + P(\la_A=1) \a[P(\,\cdot\mid\la_A=1)] = \a[P],
            \end{align}
            which is~\eqref{eq:a1-la-alpha}. Eqs.~\eqref{eq:b1-la-beta} and~\eqref{eq:c1-la-gamma} follow similarly.
        \end{proof}
    \end{claim*}

    Thus, if $\a[P]$ were $0$, as for the theoretical quantum switch data discussed in the main text, then we would have $a_1=\la_A$ with certainty whenever $\vec x=\mathbf 1$. We could then replace terms like $P(a_1\oplus b_1\oplus c_1=0\mid \vec x=\vec y=\vec z=1)$ in inequality~\eqref{eq:causal-mermin-ineq-app} with terms like $P(\la_A\oplus \la_B\oplus \la_C=0\mid \vec x=\vec y=\vec z=1)$ and derive Mermin's inequality directly.
    To generalise this to other values of $\a$ (thereby making the result robust to noise), we use the following fact about probabilities.
    \begin{lemma}
        \label{lma:probability-fact}
        For a probability distribution $P$ over a set $\Omega$ and events $A_1,A_2,\dots,A_n\subseteq\Omega$, we have
        \begin{equation}
            P(A_1)+P(A_2)+\cdots+P(A_n) \leq P(A_1, A_2,\dots, A_n) + n - 1.
        \end{equation}
        \begin{proof}
            This follows from $P(A_1, A_2) = P(A_1)+P(A_2)-P(A_1\lor A_2) \geq P(A_1)+P(A_2)-1$ for $n=2$ and induction.
        \end{proof}
    \end{lemma}

    From this we get, for instance,
    \begin{equation}
        \label{eq:halfway-there}
        \begin{split}
            P(a_1&\oplus b_3\oplus c_3 = 1 \mid \mathbf{100}) + P(a_1=\la_A\mid \mathbf{100}) \\
            &\leq P(a_1\oplus b_3\oplus c_3 = 1, a_1=\la_A \mid \mathbf{100}) + 1 \\
            &\leq P(\la_A\oplus b_3\oplus c_3 = 1 \mid \mathbf{100}) + 1 \\
            &=  P(\la_A\oplus b_3\oplus c_3 = 1 \mid \mathbf{000}) + 1,
        \end{split}
    \end{equation}
    where in the last equation we have used the relativistic causality condition~\eqref{eq:probabilistic-rc}.
    (Here we have suppressed the labels of the conditioned-upon variables, which are $\vec x$, $\vec y$, and $\vec z$, respectively.)
    Note that by condition~\eqref{eq:probabilistic-rc-extra} and the Claim above,
    \begin{equation}
        P(a_1 = \la_A \mid \mathbf{100}) = P(a_1 = \la_A \mid \vec x = \mathbf1) \geq 1-\a[P].
    \end{equation}
    This, together with~\eqref{eq:halfway-there}, implies~\eqref{eq:one-of-the-ineqs} below.
    The other inequalities below can be derived in a similar way.
    \begin{subequations}
        \begin{align}
            P(a_1\oplus b_1\oplus c_1 = 0 \mid \mathbf{111}) - \a[P] - \b[P] - \ga[P] &\leq P(\la_A\oplus \la_B\oplus \la_C = 0 \mid \mathbf{000}) \\
            P(a_1\oplus b_3\oplus c_3 = 1 \mid \mathbf{100}) - \a[P] &\leq P(\la_A\oplus b_3\oplus c_3 = 1 \mid \mathbf{000}) \label{eq:one-of-the-ineqs} \\
            P(a_3\oplus b_1\oplus c_3 = 1 \mid \mathbf{010}) - \b[P] &\leq P(a_3\oplus \la_B\oplus c_3 = 1 \mid \mathbf{000}) \\
            P(a_3\oplus b_3\oplus c_1 = 1 \mid \mathbf{001}) - \ga[P] &\leq P(a_3\oplus b_3\oplus \la_C = 1 \mid \mathbf{000})
        \end{align}
    \end{subequations}
    The desired inequality~\eqref{eq:causal-mermin-ineq-app} now follows by adding these together and noting that the sum of the right-hand sides is bounded (similarly to Mermin's inequality~\cite{Mer90a}) by 3. This last fact follows from Lemma~\ref{lma:probability-fact}, and parallels the derivation of the logical contradiction~\eqref{eq:switchy-ghz-hv-loc-cors} in the main text:
    \begin{equation}
        \label{eq:mermin-derivation}
        \begin{split}
            &\phantom{+ }P(\la_A\oplus \la_B\oplus \la_C = 0 \mid \mathbf{000}) + P(\la_A\oplus b_3\oplus c_3 = 1 \mid \mathbf{000}) \\
            &+ P(a_3\oplus \la_B\oplus c_3 = 1 \mid \mathbf{000}) + P(a_3\oplus b_3\oplus \la_C = 1 \mid \mathbf{000}) \\
            &\leq P(0 = 1 \mid \mathbf{000}) + 3 = 3.
        \end{split}
    \end{equation}
    This concludes the proof of Theorem~\ref{thm:causal-mermin}.
\end{proof}

    \section{Remark on the need for three quantum switches}
    \label{app:three-switches}
    In light of the single-switch CHSH scenario considered in~\cite{vdLBC23}, one might wonder whether a possibilistic or maximal proof of indefinite causal order under the assumption of relativistic causality is also possible in an alternative GHZ scenario that only has a quantum switch on one of the three arms of the GHZ experiment, with the other two systems simply being subjected to measurements with classical settings and outcomes.
Such a scenario is considered in Ref.~\cite{GP23}, section 2.7.2.
Whilst it is found that the probabilistic correlations observed in this scenario do exhibit indefinite causal order under the premise of relativistic causality (in the terminology of this paper), it is not maximal in the sense discussed here.
Concretely, in Ref.~\cite{GP23} it is found that the causally separable fraction is 41.2\%, rather than 0\%.

The observation that maximal nonlocality in GHZ-Mermin correlations can be carried over to maximal noncausality in the case of three quantum switches, but not in the case of one quantum switch, may be explained as follows.
The causal Mermin inequality~\eqref{eq:causal-mermin-ineq-app} for three switches was derived, in the proof in Appendix~\ref{app:mermin-ineq}, by---roughly---first realising that definite causal order implies some level of determinism of measurement outcomes~\eqref{eq:a1-la-alpha-etc}, and then using the fact that determinism and relativistic causality (read: parameter independence) together imply Mermin's inequality~\eqref{eq:mermin-derivation}.
In the presence of just one switch, only determinism of a measurement outcome in \emph{one} of the arms of the experiment can be deduced.
This amount of determinism is not sufficient to derive Mermin's inequality, even in the presence of relativistic causality.
Indeed, Mermin's inequality can be violated to the algebraic bound in a scenario where one of the three parties outputs a constant value, and the other two parties share a PR box.
(See~\eqref{eq:ghz-cors}; set $a=0$ and let $b,c$ be the outputs of a PR box with settings $y,z$.)
Crucially, PR boxes are not ruled out by the relativistic causality assumption.

This is in constrast to the CHSH inequality used in~\cite{vdLBC23}, which can be derived from relativistic causality (parameter independence) in conjunction with determinism of only one of the four measurement outcomes.

Determinism in two arms of the GHZ experiment \emph{is} sufficient to derive Mermin's inequality (cf.~\cite[Theorem 14]{HC24}), and so a setup with two quantum switches would exhibit \emph{maximal} indefinite causal order under the assumption of relativistic causality.
The \emph{possibilistic} argument developed in sections~\ref{sec:switch-data}--\ref{sec:locality} of the main text, on the other hand, intricately relies on the presence of all three quantum switches and their associated causal order variables $\la_A$, $\la_B$ and $\la_C$.

Appendix~\ref{app:chained-ineqs} below provides a certification which (in a limit) is maximal while using only one switch; this is based not on the GHZ experiment, but on a setup more akin to the one used in Ref.~\cite{vdLBC23}.

    \section{Chained inequalities}
    \label{app:chained-ineqs}
That the Mermin inequality (resp.\ causal Mermin inequality~\eqref{eq:causal-mermin-ineq}) is violated up to the algebraic maximum (i.e.\ the maximum value among all valid probability distributions) implies that no fraction of the distribution allows a local deterministic hidden variable model (resp.\ definite causal order model satisfying relativistic causality).
However, this only implies that it is never the case that all of the three switches have a causal order simultaneously.
Here we show that a similar result can be derived for a single quantum switch, by using the idea of chained Bell inequalities introduced by Braunstein and Caves~\cite{BC90}, which are violated to the algebraic maximum in the limit of a large number of measurement settings $N$.

We first recall Braunstein and Caves's chained inequalities as they apply to Bell nonlocality.
Consider a scenario with two spacelike-separated parties $\pA$ and $\pB$, having inputs $x,y\in\{0,1,\dots,N-1\}$ and outputs $a,b\in\{0,1\}$.
Any choice of a pair of values for $x$ as well as for $y$ can be used to define a Clauser-Horne-Shimony-Holt (CHSH) expression: more precisely, for any probability distribution $R(a b | x y)$ and $\xi,\xi',\upsilon,\upsilon'\in\{0,1,\dots,N-1\}$, define
\begin{equation}
    \begin{split}
        \CHSH_{\xi,\xi';\upsilon,\upsilon'}[R] &\coloneqq
        R(a = b \mid x=\xi, y=\upsilon) \\
        &+ R(a = b \mid x=\xi', y=\upsilon) \\
        &+ R(a = b \mid x=\xi', y=\upsilon') \\
        &- R(a = b \mid x=\xi, y=\upsilon'). \\
    \end{split}
\end{equation}
With this convention for the CHSH expression, the classical bound is 2, the quantum bound $1+\sqrt 2$, and the algebraic bound is 3.
Summing an appropriate set of these CHSH expressions together defines the Braunstein-Caves expression
\begin{equation}
    \BC_N[R] \coloneqq \sum_{i=0}^{N-2} \CHSH_{0,i+1;i,i+1}[R] \label{eq:bc-as-chsh-sum}
\end{equation}
which after some rewriting becomes
\begin{multline}
    \BC_N[R] = \sum_{i=0}^{N-2} \left[ R(a = b \mid x=i,y=i) + R(a = b \mid x=i+1,y=i) \right] \\
    + R(a=b|x=N-1,y=N-1) - R(a=b|x=0,y=N-1). \label{eq:bc-expanded}
\end{multline}

With the classical bound for each of the CHSH expressions being 2, the overall classical bound on $\BC_N[R]$ can be read off from~\eqref{eq:bc-as-chsh-sum} to be
\begin{equation}
    \BC_N[R] \leq 2N-2.
\end{equation}
The algebraic maximum of $\BC_N[R]$, on the other hand, can seen from~\eqref{eq:bc-expanded} to be $2N-1$.

A high quantum value for $\BC_N$ can be achieved by letting $\pA$ and $\pB$ measure the two parts of the Bell state $\ket{\Phi^+}$ in the measurement directions shown in Figure~\ref{fig:bc-directions}.
The directions are spaced an angle $\theta = \pi/2N$ apart on the Bloch sphere.
All terms in the expansion~\eqref{eq:bc-expanded} except the last involve a neighbouring pair of measurement directions, yielding probabilities of $\cos^2 \theta/2$. In total, the quantum value is
\begin{equation}
    \label{eq:bc-vanilla-quantum-value}
    \begin{split}
        \BC_N[R] &= (2N-1) \cos^2\frac{\theta}{2} - \left(1-\cos^2\frac{\theta}{2}\right) \\
        &= 2N \cos^2\frac{\theta}{2} - 1
        = N\left(\cos\frac{\pi}{2N} +1\right) - 1 \quad\xrightarrow{N\gg 1}\quad 2N-1,
    \end{split}
\end{equation}
which approaches the algebraic maximum in the limit of large $N$.

A crucial observation is that the classical bound of 2 for a CHSH expressions follows (under additional requirements sometimes called measurement independence and parameter independence) not only from the assumption that all measurements outcomes are predetermined, but also from the weaker requirement that only $\pA$'s outcome $a$ for one of her inputs $x$ is predetermined.
Because all CHSH expressions in~\eqref{eq:bc-as-chsh-sum} involve the input $x=0$, this means that the classical bound $\BC_N[R]\leq 2N-2$ also holds when only the $x=0$ measurement is predetermined.
This will allow us to make the translation from Bell nonlocality to indefinite causal order in the quantum switch, by using a quantum switch to simulate the $x=0$ measurement (in the spirit of Observation~\ref{fct:x-are-0} in the main text), while noting that the definite causal order assumption implies predetermination of that measurement (in the spirit of Lemma~\ref{lma}).

The relevant quantum switch scenario is depicted in Figure~\ref{fig:switchy-bc}.
The quantum switch is now controlled in the $Z$ basis (rather than $X$, as in the main text).
$\pA_1$ and $\pA_2$ perform the same quantum instruments as described in the main text and Appendix~\ref{app:switch-data}.
$\pA_3$ and $\pB$ perform measurements on the output control qubit and a qubit $B$ entangled to the input control qubit, respectively.
Their inputs $x_3,y$ take values in $\{0,1,\dots,N\}$ and their measurement directions are identical to those in the usual chained inequality scenario just described (Figure~\ref{fig:bc-directions}).
All other variables are binary, as in the main text.
This setup defines a conditional probability distribution $Q(a_1 a_2 a_3 b | x_1 x_2 x_3 y)$, which in turn defines an unconditional probability distribution $Q(a_1 a_2 a_3 b x_1 x_2 x_3 y)$ by assuming uniform distribution of the inputs: $Q(x_1 x_2 x_3 y) = 1/4N^2$.

In the main text, our first step was to find a subset of the quantum switch data that exhibited GHZ-like correlations (i.e.\ going from~\eqref{eq:ghz-cors} to~\eqref{eq:switchy-ghz-cors}).
Likewise, here we will need to find a subset of the data $Q$ that behaves as $R$ in~\eqref{eq:bc-vanilla-quantum-value} above, i.e.\ which violates the Braunstein-Caves inequality to the algebraic maximum in the limit of large $N$.
To do this, note that one can consider $\cA_3$ and $\cB$ as performing a Braunstein-Caves test on the shared entangled state $\ket{\Phi^+}$, with two differences.
First of all, when $x_3=0$, we need to ensure that $x_1=x_2=1$ and let $\cA_3$ output $a_1$, as that is a simulation of the computational basis measurement (Observation~\ref{fct:x-are-1}).
When $x_3\neq 0$, on the other hand, we need to ensure that $x_1=x_2=0$, so that the outcome $a_3$ is not disturbed by the presence of the quantum switch (Observation~\ref{fct:x-are-0}).
The data in $Q$ that are relevant to the derivation and violation of Braunstein-Caves inequalities is therefore summarised in the distribution $R_Q$ defined in~\eqref{eq:r_p} below.

\begin{figure}
    \centering

    \begin{minipage}{.45\textwidth}
        \centering
        \vspace{5em}
        \begin{tikzpicture}
	\begin{pgfonlayer}{nodelayer}
		\node [style=none] (0) at (0, 0) {};
		\node [style=none] (1) at (0, 3) {};
		\node [style=none] (2) at (3, 0) {};
		\node [style=none] (3) at (0, -3) {};
		\node [style=none] (4) at (-3, 0) {};
		\node [style=none] (5) at (1.25, 2.75) {};
		\node [style=none] (6) at (2.25, 2) {};
		\node [style=none] (7) at (2.85, 0.975) {};
		\node [style=none] (9) at (1.25, -2.75) {};
		\node [style=none] (10) at (2.25, -2) {};
		\node [style=none] (11) at (2.275, -0.175) {\rotatebox{347}{\vdots}};
		\node [style=none] (12) at (0, 1.25) {};
		\node [style=none] (13) at (0.525, 1.125) {};
		\node [style=none] (14) at (0.35, 1.75) {$\theta$};
		\node [style=none] (15) at (0.5, 3.5) {\color{blue}\footnotesize$x=0$};
		\node [style=none] (16) at (2.025, 3) {\color{olive}\footnotesize$y=0$};
		\node [style=none] (17) at (3.1, 2.15) {\color{blue}\footnotesize$x=1$};
		\node [style=none] (18) at (3.65, 1.175) {\color{olive}\footnotesize$y=1$};
		\node [style=none] (19) at (3.575, -2.2) {\color{blue}\footnotesize$x=N-1$};
		\node [style=none] (20) at (2.65, -2.95) {\color{olive}\footnotesize$y=N-1$};
		\node [style=none] (21) at (-5, -1) {};
		\node [style=none] (22) at (-5, -3) {};
		\node [style=none] (23) at (-3.25, -3) {};
		\node [style=none] (24) at (-4.5, -1.25) {$Z$};
		\node [style=none] (25) at (-3.5, -3.5) {$X$};
	\end{pgfonlayer}
	\begin{pgfonlayer}{edgelayer}
		\draw [style=dik] (1.center)
			 to [bend right=45] (4.center)
			 to [bend right=45] (3.center)
			 to [bend right=45] (2.center)
			 to [bend right=45] cycle;
		\draw [style=pijl blauw] (0.center) to (1.center);
		\draw [style=pijl blauw] (0.center) to (6.center);
		\draw [style=pijl groen] (0.center) to (5.center);
		\draw [style=pijl groen] (0.center) to (7.center);
		\draw [style=pijl groen] (0.center) to (9.center);
		\draw [style=pijl blauw] (0.center) to (10.center);
		\draw [style=grey] (0.center) to (3.center);
		\draw [style=dik, bend left, looseness=0.75] (12.center) to (13.center);
		\draw [style=arrow] (22.center) to (21.center);
		\draw [style=arrow] (22.center) to (23.center);
	\end{pgfonlayer}
\end{tikzpicture}
        \vspace{4.9em}
        \caption{Directions of qubit measurements in the $Z$-$X$ plane that yield a value of the Braunstein-Caves quantity $\BC_N$ converging to the algebraic maximum.}
        \label{fig:bc-directions}
    \end{minipage}%
    \hspace{.1\textwidth}%
    \begin{minipage}{.45\textwidth}
        \centering
        \begin{tikzpicture}
	\begin{pgfonlayer}{nodelayer}
		\node [style=none] (54) at (-5.5, 1) {};
		\node [style=none] (60) at (-6.5, 1) {};
		\node [style=none] (61) at (-3.5, 1) {};
		\node [style=none] (62) at (-5, -0.5) {};
		\node [style=none] (123) at (-4.5, 1) {};
		\node [style=none] (142) at (-4.95, 0.475) {$\Phi^+$};
		\node [style=none] (143) at (-11.25, 9.5) {};
		\node [style=none] (144) at (-11.25, 8.25) {};
		\node [style=none] (145) at (-9.25, 8.25) {};
		\node [style=none] (146) at (-7.75, 8.25) {};
		\node [style=none] (147) at (-5.75, 8.25) {};
		\node [style=none] (148) at (-5.75, 9.5) {};
		\node [style=none] (149) at (-9.25, 6.25) {};
		\node [style=none] (150) at (-11.25, 6.25) {};
		\node [style=none] (151) at (-11.25, 5) {};
		\node [style=none] (152) at (-7.75, 6.25) {};
		\node [style=none] (153) at (-5.75, 6.25) {};
		\node [style=none] (154) at (-5.75, 5) {};
		\node [style=none] (155) at (-10.25, 7.825) {};
		\node [style=none] (156) at (-10.25, 6.675) {};
		\node [style=none] (159) at (-10.25, 8.25) {};
		\node [style=none] (160) at (-10.25, 6.25) {};
		\node [style=none] (163) at (-7.75, 9.5) {};
		\node [style=none] (164) at (-7.75, 11.25) {};
		\node [style=none] (165) at (-9.35, 9.5) {};
		\node [style=none] (166) at (-9.325, 11.175) {};
		\node [style=none] (167) at (-9.25, 4) {};
		\node [style=none] (168) at (-7.5, 5) {};
		\node [style=none] (170) at (-7.5, 4) {};
		\node [style=none] (171) at (-7.2, 10.5) {$T_A$};
		\node [style=none] (172) at (-8.75, 10.5) {$C_A$};
		\node [style=none] (173) at (-9.75, 4.425) {$C_A$};
		\node [style=none] (174) at (-6.95, 4.45) {$T_A$};
		\node [style=none] (175) at (-8.5, 8.8) {\small$\textsc{switch}$};
		\node [style=grijs] (190) at (-12.375, 5.75) {$x_1$};
		\node [style=grijs] (193) at (-12.35, 8.75) {$a_1$};
		\node [style=none] (214) at (-7.75, 11.25) {};
		\node [style=none] (215) at (-8.275, 11.25) {};
		\node [style=none] (216) at (-7.225, 11.25) {};
		\node [style=none] (217) at (-8.1, 11.5) {};
		\node [style=none] (218) at (-7.375, 11.5) {};
		\node [style=none] (219) at (-7.9, 11.75) {};
		\node [style=none] (220) at (-7.6, 11.75) {};
		\node [style=none] (227) at (-8.25, 4) {};
		\node [style=none] (228) at (-6.75, 4) {};
		\node [style=none] (229) at (-7.5, 3) {};
		\node [style=none] (230) at (-7.5, 3.6) {0};
		\node [style=none] (239) at (-10.575, 11.175) {};
		\node [style=none] (240) at (-8.825, 11.175) {};
		\node [style=grijs] (244) at (-10.225, 13.125) {$a_3$};
		\node [style=none] (406) at (-12.25, 7.825) {};
		\node [style=none] (407) at (-9.75, 7.825) {};
		\node [style=none] (408) at (-9.75, 6.675) {};
		\node [style=none] (409) at (-12.25, 6.675) {};
		\node [style=none] (410) at (-11.85, 7.825) {};
		\node [style=none] (411) at (-11.85, 6.675) {};
		\node [style=none] (412) at (-11.85, 8.75) {};
		\node [style=none] (413) at (-11.85, 5.75) {};
		\node [style=none] (414) at (-10.575, 12.125) {};
		\node [style=none] (415) at (-8.85, 12.125) {};
		\node [style=none] (416) at (-9.7, 12.125) {};
		\node [style=none] (417) at (-9.7, 13.125) {};
		\node [style=none] (418) at (-9.7, 11.65) {$\mathcal{A}_3$};
		\node [style=none] (419) at (-11, 7.25) {$\mathcal{A}_1$};
		\node [style=none] (454) at (-6.75, 7.825) {};
		\node [style=none] (455) at (-6.75, 6.675) {};
		\node [style=none] (456) at (-6.75, 8.25) {};
		\node [style=none] (457) at (-6.75, 6.25) {};
		\node [style=grijs] (458) at (-4.625, 5.75) {$x_2$};
		\node [style=grijs] (459) at (-4.65, 8.75) {$a_2$};
		\node [style=none] (460) at (-4.75, 7.825) {};
		\node [style=none] (461) at (-7.25, 7.825) {};
		\node [style=none] (462) at (-7.25, 6.675) {};
		\node [style=none] (463) at (-4.75, 6.675) {};
		\node [style=none] (464) at (-5.15, 7.825) {};
		\node [style=none] (465) at (-5.15, 6.675) {};
		\node [style=none] (466) at (-5.15, 8.75) {};
		\node [style=none] (467) at (-5.15, 5.75) {};
		\node [style=none] (468) at (-6, 7.25) {$\mathcal{A}_2$};
		\node [style=none] (469) at (-1, 6.5) {};
		\node [style=none] (484) at (-10.075, 10.175) {};
		\node [style=none] (485) at (-10.075, 11.175) {};
		\node [style=grijs] (486) at (-10.6, 10.175) {$x_3$};
		\node [style=none] (488) at (-1.7, 5.475) {$B$};
		\node [style=none] (489) at (0.225, 6.5) {};
		\node [style=none] (490) at (-1.525, 6.5) {};
		\node [style=grijs] (491) at (-0.275, 8.45) {$b$};
		\node [style=none] (492) at (0.225, 7.45) {};
		\node [style=none] (493) at (-1.5, 7.45) {};
		\node [style=none] (494) at (-0.65, 7.45) {};
		\node [style=none] (495) at (-0.65, 8.45) {};
		\node [style=none] (496) at (-0.65, 6.975) {$\mathcal{B}$};
		\node [style=none] (497) at (-0.275, 5.5) {};
		\node [style=none] (498) at (-0.275, 6.5) {};
		\node [style=grijs] (499) at (0.0999999, 5.5) {$y$};
		\node [style=none] (500) at (-11.75, 13.25) {};
		\node [style=none] (501) at (-11.75, 14.25) {};
		\node [style=none] (502) at (-10.75, 13.25) {};
		\node [style=none] (503) at (-11.75, 12.25) {};
		\node [style=none] (504) at (-12.75, 13.25) {};
		\node [style=none] (505) at (-10.875, 13.75) {};
		\node [style=none] (506) at (-11.225, 14.125) {};
		\node [style=none] (507) at (-11.25, 12.375) {};
		\node [style=none] (508) at (-10.875, 12.75) {};
		\node [style=none] (509) at (-2.275, 8.75) {};
		\node [style=none] (510) at (-2.275, 9.75) {};
		\node [style=none] (511) at (-1.275, 8.75) {};
		\node [style=none] (512) at (-2.275, 7.75) {};
		\node [style=none] (513) at (-3.275, 8.75) {};
		\node [style=none] (514) at (-1.4, 9.25) {};
		\node [style=none] (515) at (-1.75, 9.625) {};
		\node [style=none] (516) at (-1.775, 7.875) {};
		\node [style=none] (517) at (-1.4, 8.25) {};
		\node [style=none] (518) at (-9.25, 5) {};
		\node [style=none] (519) at (-11.1, 13.225) {\scalebox{.5}{\rotatebox{270}{\dots}}};
		\node [style=none] (520) at (-1.675, 8.4) {\scalebox{.5}{\rotatebox{65}{\dots}}};
	\end{pgfonlayer}
	\begin{pgfonlayer}{edgelayer}
		\draw [style=grey dashed] (417.center) to (416.center);
		\draw [style=grey dashed] (485.center) to (484.center);
		\draw [style=dik] (61.center)
			 to (62.center)
			 to (60.center)
			 to cycle;
		\draw [style=dik filled] (146.center)
			 to (147.center)
			 to (148.center)
			 to (143.center)
			 to (144.center)
			 to (145.center)
			 to (149.center)
			 to (150.center)
			 to (151.center)
			 to (154.center)
			 to (153.center)
			 to (152.center)
			 to cycle;
		\draw [style=dik] (159.center) to (155.center);
		\draw [style=dik] (156.center) to (160.center);
		\draw [style=dik] (164.center) to (163.center);
		\draw [style=dik] (165.center) to (166.center);
		\draw [style=dik, in=90, out=-90] (168.center) to (170.center);
		\draw [style=dik] (219.center) to (220.center);
		\draw [style=dik] (217.center) to (218.center);
		\draw [style=dik] (215.center) to (216.center);
		\draw [style=dik] (227.center) to (228.center);
		\draw [style=dik] (228.center) to (229.center);
		\draw [style=dik] (229.center) to (227.center);
		\draw [style=dik] (240.center) to (239.center);
		\draw [style=dik, in=90, out=-90] (167.center) to (54.center);
		\draw [style=dik] (409.center) to (406.center);
		\draw [style=dik] (406.center) to (407.center);
		\draw [style=dik] (407.center) to (408.center);
		\draw [style=dik] (408.center) to (409.center);
		\draw [style=grey dashed] (412.center) to (410.center);
		\draw [style=grey dashed] (411.center) to (413.center);
		\draw [style=dik] (414.center) to (415.center);
		\draw [style=dik] (415.center) to (240.center);
		\draw [style=dik] (414.center) to (239.center);
		\draw [style=dik] (456.center) to (454.center);
		\draw [style=dik] (455.center) to (457.center);
		\draw [style=dik] (463.center) to (460.center);
		\draw [style=dik] (460.center) to (461.center);
		\draw [style=dik] (461.center) to (462.center);
		\draw [style=dik] (462.center) to (463.center);
		\draw [style=grey dashed] (466.center) to (464.center);
		\draw [style=grey dashed] (465.center) to (467.center);
		\draw [style=dik, in=90, out=-90] (469.center) to (123.center);
		\draw [style=grey dashed] (495.center) to (494.center);
		\draw [style=grey dashed] (498.center) to (497.center);
		\draw [style=dik] (490.center) to (489.center);
		\draw [style=dik] (492.center) to (493.center);
		\draw [style=dik] (493.center) to (490.center);
		\draw [style=dik] (492.center) to (489.center);
		\draw [style=dik, bend right=45] (501.center) to (504.center);
		\draw [style=dik, bend left=45] (502.center) to (503.center);
		\draw [style=dik, bend left=45] (501.center) to (502.center);
		\draw [style=dik, bend left=45] (503.center) to (504.center);
		\draw [style=pijl blauw] (500.center) to (501.center);
		\draw [style=pijl blauw] (500.center) to (505.center);
		\draw [style=pijl blauw] (500.center) to (508.center);
		\draw [style=dik, bend right=45] (510.center) to (513.center);
		\draw [style=dik, bend left=45] (511.center) to (512.center);
		\draw [style=dik, bend left=45] (510.center) to (511.center);
		\draw [style=dik, bend left=45] (512.center) to (513.center);
		\draw [style=pijl groen] (509.center) to (515.center);
		\draw [style=pijl groen] (509.center) to (511.center);
		\draw [style=pijl groen] (509.center) to (516.center);
		\draw [style=dik] (518.center) to (167.center);
	\end{pgfonlayer}
\end{tikzpicture}
        \caption{Quantum switch setup producing correlations that violate the causal Braunstein-Caves inequality~\eqref{eq:switchy-bc-ineq} to the algebraic maximum in the limit of large $N$. The measurement directions of $\cA_3$ and $\cB$ are as in Figure~\ref{fig:bc-directions}.}
        \label{fig:switchy-bc}
    \end{minipage}
\end{figure}

We then get the following theorem. As in the main text, the conditions~\eqref{eq:chained-conditions} are motivated by the assumptions of free interventions, definite causal order, and relativistic causality.

\begin{theorem}
    \label{thm:causal-bc-app}
    Let $\la$ be a random variable with domain $\{0,1\}$.
    Suppose the probability distribution $P(a_1 a_2 a_3 b x_1 x_2 x_3 y \la)$ satisfies
    \begin{subequations}
        \label{eq:chained-conditions}
        \begin{align}
            \forall x_1 x_2 x_3 y: P(x_1 x_2 & x_3 y)  >0; \label{eq:bc-unif-settings}\\
            \la \indep_P x_1 x_2 & x_3 y ; \\
            a_1\indep_P x_2 \mid \la=0;  \qquad & a_2\indep_P x_1 \mid \la=1; \label{eq:bc-dco} \\
            \la b \indep_P x_1 x_2 x_3 \mid y; \qquad \la a_1 a_2 a_3 \indep_P y &\mid x_1 x_2 x_3; \qquad \la a_1 a_2 \indep_P x_3 y \mid x_1 x_2,  \label{eq:bc-rc}
        \end{align}
    \end{subequations}
    and let
    \begin{equation}
        \label{eq:r_p}
        R_P(a b \la \mid x_3 y) \coloneqq
        \begin{cases}
            P(a_1 = a, b, \la \mid x_1=x_2=1, x_3, y) &\quad\text{if } x_3 = 0; \\
            P(a_3 = a, b, \la \mid x_1=x_2=0, x_3, y) &\quad\text{if } x_3\neq 0
        \end{cases}
    \end{equation}
    and
    \begin{equation}
        \a[P] \coloneqq P(a_1=1 \mid x_1 x_2=10) + P(a_2=1 \mid x_1 x_2=01) + P(a_1 a_2=00 \mid x_1 x_2=11).
    \end{equation}
    Then $P$ satisfies the inequality
    \begin{equation}
        \label{eq:switchy-bc-ineq}
        \BC_N[R_P] - (2N-2) \alpha[P] \leq 2N-2.
    \end{equation}
    The quantum value of this expression is
    \begin{equation}
        \label{eq:bc-switch-quantum-value}
        \BC_N[R_Q] - (2N-2) \alpha[Q] = N\left(\cos\frac{\pi}{2N} +1\right) - 1 \quad\xrightarrow{N\gg 1}\quad 2N-1.
    \end{equation}
\end{theorem}
\begin{proof}
    $R_P$ satisfies measurement and parameter independence, i.e.\ the conditions
    \begin{equation}
        \la \indep_{R_P} x_3 y, \quad a \indep_{R_P} y \mid x_3\la,\quad\text{ and }\quad b\indep_{R_P} x_3 \mid y\la.
    \end{equation}
    This means we can use the monogamy inequality of Ref.~\cite{BKP06} to obtain
    \begin{equation}
        \CHSH_{0,i+1;i,i+1}[R_P] \leq 4 - 2 R_P(a=\la|x_3=0) \quad\text{ for any } i\in\{0,1,\dots,N-1\}.
    \end{equation}
    But by definition of $R_P$ and the definite causal order assumption~\eqref{eq:bc-dco}, we get, via an argument similar to that for~\eqref{eq:a1-la-alpha},
    \begin{equation}
        R_P(a = \la | x_3=0) = P(a_1 = \la | x_1=x_2=1) \geq 1-\alpha[P].
    \end{equation}
    Thus
    \begin{equation}
        \CHSH_{0,i+1;i,i+1}[R_P] - 2\a[P] \leq 2,
    \end{equation}
    and summing over $i$ implies the sought-after inequality~\eqref{eq:switchy-bc-ineq}.

    The quantum switch correlations $Q$, on the other hand, yield---by virtue of versions of Observations~\ref{fct:x-are-0} and~\ref{fct:x-are-1} or Lemmas~\ref{lma:obs-1} and~\ref{lma:obs-2} in Appendix~\ref{app:switch-data}---a value of $\BC_N[R_Q]$ identical to that given in~\eqref{eq:bc-vanilla-quantum-value}, and a value of 0 for $\alpha[Q]$, thus giving~\eqref{eq:bc-switch-quantum-value}.
\end{proof}

\begin{corollary}
    Suppose $Q(a_1 a_2 a_3 b x_1 x_2 x_3 y)$ satisfies~\eqref{eq:bc-switch-quantum-value} and also admits a hidden variable model $Q(a_1 a_2 a_3 b x_1 x_2 x_3 y \la)$ satisfying~\eqref{eq:chained-conditions}, but where $\la$ now ranges over $\{0,1,\bot\}$.
    Here $\la = \bot$ denotes the absence of a (relativistically well-behaved) definite causal order.
    Then the fraction of runs with a definite causal order is
    \begin{equation}
        Q(\la\in\{0,1\}) \leq 2N - N\left( \cos{\frac{\pi}{2N}} + 1 \right) \to 0 \quad\text{ as }\quad N\to\infty.
    \end{equation}
\end{corollary}
\begin{proof}
    Write $Q$ as the convex sum $Q = Q(\la\in\{0,1\}) Q^{\{0,1\}} + Q(\la = \bot) Q^\bot$, where
    \begin{align}
        Q^{\{0,1\}}(a_1 a_2 a_3 b x_1 x_2 x_3 y \la) &\coloneqq Q(a_1 a_2 a_3 b x_1 x_2 x_3 y \la \mid \la \in\{0,1\}) \\
        \text{and\quad} Q^\bot(a_1 a_2 a_3 b x_1 x_2 x_3 y \la) &\coloneqq Q(a_1 a_2 a_3 b x_1 x_2 x_3 y \la \mid \la =\bot).
    \end{align}
    Theorem~\ref{thm:causal-bc-app} gives
    \begin{equation}
        \BC_N\left[R_{Q^{\{0,1\}}}\right] - (2N-2) \operatorname{\alpha}\left[Q^{\{0,1\}}\right] \leq 2N-2.
    \end{equation}
    Because the algebraic maximum of $\BC_N$ is $2N-1$ and the inequality is linear, this implies
    \begin{align}
        \BC_N\left[R_Q\right] - (2N-2) \operatorname{\alpha}\left[Q\right] &\leq Q(\la\in\{0,1\}) (2N-2) + Q(\la=\bot) (2N-1) \\
        &= 2N - 1 - Q(\la \in\{0,1\}).
    \end{align}
    (Note that in the extreme case where one imposes that a definite causal order exists on each run, i.e.\ $Q(\la\in\{0,1\})=1$, the bound in~\eqref{eq:switchy-bc-ineq} is recovered.)
    Inserting~\eqref{eq:bc-switch-quantum-value} and rearranging now gives the result.
\end{proof}

\end{document}